\documentclass[a4paper,aps,prl,10pt,twocolumn,floatfix,superscriptaddress,notitlepage,usenames,dvipsnames,
svgnames,table,longbibliography]{revtex4-2}
\pdfoutput=1
\usepackage[utf8]{inputenc}
\usepackage[T1]{fontenc}
\usepackage{graphicx}
\usepackage{grffile}
\usepackage{wrapfig}
\usepackage{rotating}
\usepackage[normalem]{ulem}
\usepackage{amsmath}
\usepackage{textcomp}
\usepackage{amssymb}
\usepackage{capt-of}
\usepackage{soul}
\usepackage{parskip}
\usepackage{mathrsfs}
\usepackage[margin= 0.75in]{geometry}
\usepackage[braket, qm]{qcircuit}
\usepackage{footmisc}

\usepackage{pgfplots}
\definecolor{blue-violet}{rgb}{0.54, 0.17, 0.89}

\usepackage{bm}
\usepackage{dsfont}
\usepackage[font=small,labelfont=bf,justification=RaggedRight,format=plain]{caption}
\usepackage{subcaption}
\usepackage{enumerate}
\usepackage{hhline}
\usepackage{bbm}
\usepackage{etoolbox}
\usepackage{tcolorbox} 
\usepackage{amsthm}
\usepackage{amssymb}
\usepackage{mathtools}
\usepackage{thmtools}
\usepackage{thm-restate}

\newcommand{\prlsection}[1]{{\em {#1}.---~}}
\usepackage{microtype}
\usepackage{dsfont}
\usepackage{tensor}
\usepackage{physics}
\usepackage{mathrsfs}
\usepackage{mathtools}
\usepackage{fixltx2e}
\usepackage{relsize}

\DeclarePairedDelimiterX\phys[2]{\langle}{\rangle}{#1 \delimsize\vert\mathopen{} #2}

\newtheorem{thm}{Theorem}

\newtheorem{proposition}{Proposition}
\newtheorem{lemma}{Lemma}
\newtheorem{cor}{Corollary}
\newtheorem{obs}{Observation}

\theoremstyle{remark}

\usepackage{upgreek}

\newcounter{mysubequations}

\usepackage{tikz}
\usepackage{tikz-network}
\usetikzlibrary{backgrounds,patterns,decorations.pathreplacing,decorations.pathmorphing,patterns.meta,arrows.meta}
\usepackage{color}
\definecolor{myred}{RGB}{232,102,102}
\definecolor{myblue}{RGB}{187,187,255}
\definecolor{mybeige}{RGB}{245,245,220}
\definecolor{myorange}{RGB}{255,165,0}
\definecolor{mygrey}{RGB}{105,105,105}
\definecolor{OliveGreen}{RGB}{85,107,47}
\definecolor{NavyBlue}{RGB}{0,0,128}
\definecolor{mygreen}{RGB}{34,139,34}
\definecolor{myY}{RGB}{220,255,203}
\definecolor{myYO}{RGB}{255, 220, 151}
\definecolor{mygreenc}{RGB}{150,50,50}

\newcommand{\Agate}[2]{
\draw[thick] (#1, #2 + 0.6) -- (#1,#2-0.6);
\draw[ thick, fill=myblue, rounded corners=2pt] (#1-0.35,#2+0.35) rectangle (#1+0.35,#2-0.35);
\draw[thick] (#1+0.05,#2+0.25) -- (#1+0.25,#2+0.25) -- (#1+0.25,#2+0.05);
\node at (#1,#2) {$A$};
}

\newcommand{\Adaggate}[2]{
\draw[thick] (#1, #2 + 0.6) -- (#1,#2-0.6);
\draw[ thick, fill=myblue, rounded corners=2pt] (#1-0.35,#2+0.35) rectangle (#1+0.35,#2-0.35);
\draw[thick] (#1+0.05,#2+0.25) -- (#1+0.25,#2+0.25) -- (#1+0.25,#2+0.05);
\node at (#1,#2) {$A^\dagger$};
}

\newcommand{\Pigate}[2]{
\draw[thick] (#1, #2 + 0.6) -- (#1,#2-0.6);
\draw[thick, fill=mybeige, rounded corners=2pt] (#1-0.35,#2+0.35) rectangle (#1+0.35,#2-0.35);
\draw[thick] (#1+0.05,#2+0.25) -- (#1+0.25,#2+0.25) -- (#1+0.25,#2+0.05);
\node at (#1,#2) {$\Pi$};
}

\newcommand{\sigmagate}[2]{
\draw[thick] (#1, #2 + 0.6) -- (#1,#2-0.6);
\draw[ thick, fill=mybeige, rounded corners=2pt] (#1-0.35,#2+0.35) rectangle (#1+0.35,#2-0.35);
\draw[thick] (#1+0.05,#2+0.25) -- (#1+0.25,#2+0.25) -- (#1+0.25,#2+0.05);
\node at (#1,#2) {$\sigma_\alpha$};
}

\newcommand\doubleds{0.7}

\newcommand{\TrLeft}[1]{
	\begin{scope}[shift={(#1)}]
      \draw [thick] (0,0) to (-\doubleds/2,0);
	   \draw [thick] (-\doubleds/2,0) to  [bend left=90] (-\doubleds/2,\doubleds);
	   \draw [thick, dotted] (-\doubleds/2,\doubleds) to  (\doubleds/2,\doubleds);
	\end{scope}
}
\newcommand{\TrRight}[1]{
	\begin{scope}[shift={(#1)}, xscale=-1]
	   \TrLeft{(0,0)};
	\end{scope}
}

\usepackage{appendix}
\usepackage{hyperref}
\hypersetup{
 pdfauthor={Faidon Andreadakis, and Paolo Zanardi},
 pdftitle={An Exact Link between Nonlocal Nonstabilizerness and Operator Entanglement},
 pdflang={English},colorlinks=true,linkcolor=RubineRed,citecolor=blue-violet,urlcolor=Cerulean} 
\usepackage[capitalise]{cleveref}
\crefname{exmp}{Example}{Examples} 
\crefname{thm}{Theorem}{Theorems}
\crefname{obs}{Observation}{Observations}
\crefname{cor}{Corollary}{Corollaries}
\Crefname{appsec}{appendix}{appendices}

\begin{document}

\title{An Exact Link between Nonlocal Nonstabilizerness and Operator Entanglement}

\author{Faidon Andreadakis}
\email [e-mail: ]{fandread@usc.edu}
\affiliation{Department of Physics and Astronomy, and Center for Quantum Information Science and Technology, University of Southern California, Los Angeles, California 90089-0484, USA}

\author{Paolo Zanardi}
\email [e-mail: ]{zanardi@usc.edu}
\affiliation{Department of Physics and Astronomy, and Center for Quantum Information Science and Technology, University of Southern California, Los Angeles, California 90089-0484, USA}
\affiliation{Department of Mathematics, University of Southern California, Los Angeles, California 90089-2532, USA}

\date{\today}

\begin{abstract}
Nonstabilizerness is a quantum property of states associated with the non-Clifford resources required for their preparation. As a resource, nonstabilizerness complements entanglement, and the interplay between these two concepts has garnered significant attention in recent years. In this work, we establish an exact correspondence between the generation of nonlocal nonstabilizerness and operator entanglement under unitary evolutions. Nonlocal nonstabilizerness refers to nonstabilizerness that cannot be erased via local operations, while operator entanglement generalizes entanglement to operator space, characterizing the complexity of operators across a bipartition. Specifically, we prove that a unitary map generates nonlocal nonstabilizerness if and only if it generates operator entanglement on Pauli strings. Guided by this result, we introduce an average measure of a unitary's Pauli-entangling power, serving as a proxy for nonlocal nonstabilizerness generation. We derive analytical formulas for this measure and examine its properties, including its typical value and upper bounds in terms of the nonstabilizerness properties of the evolution.
\end{abstract}
\maketitle

\prlsection{Introduction}
A defining feature of quantum dynamics is the emergence of inherently non-classical properties. A key example is the generation of entanglement \cite{horodeckiQuantumEntanglement2009}, which manifests as the buildup of non-classical, non-local correlations across the system's degrees of freedom. Entanglement plays a central role in the study of quantum many-body systems \cite{amicoEntanglementManybodySystems2008}, controlling the effectiveness of tensor network methods for simulating quantum many-body states \cite{hastingsAreaLawOnedimensional2007,verstraeteMatrixProductStates2008,verstraeteMatrixProductStates2006,
schuchEntropyScalingSimulability2008} and serving as a witness of quantum phase transitions \cite{vidalEntanglementQuantumCritical2003,
osborneEntanglementSimpleQuantum2002,calabreseEntanglementEntropyQuantum2004}. Yet, the generation of entanglement alone is not sufficient to render quantum dynamics classically intractable. This follows immediately from the fact that Clifford unitaries are classically simulable \cite{gottesmanHeisenbergRepresentationQuantum1998,
aaronsonImprovedSimulationStabilizer2004,andersFastSimulationStabilizer2006}, while also, generically, generating near maximal entanglement \cite{nahumQuantumEntanglementGrowth2017,
lindenEntanglingDisentanglingPower2009,zanardiEntanglingPowerQuantum2000}. The above observation is directly related to the theory of nonstabilizerness \cite{
veitchResourceTheoryStabilizer2014,leoneStabilizerRenyiEntropy2022}. Nonstabilizerness is generated by non-Clifford unitaries and is a necessary complementary resource to entanglement to achieve universal quantum computation \cite{bravyiUniversalQuantumComputation2005,
harrowQuantumComputationalSupremacy2017}.
\par 
Broadly speaking, nonstabilizerness and entanglement are distinct resources. Clifford unitaries can generate entanglement without introducing nonstabilizerness, while local non-Clifford unitaries can generate nonstabilizerness without producing entanglement. However, in the absence of entanglement, only a restricted form of nonstabilizerness is accessible \cite{olivieroMagicstateResourceTheory2022,
odavicStabilizerEntropyNonintegrable2025}. Further exploring the interplay between the nonstabilizerness and entanglement of quantum states has been the subject of several recent investigations \cite{guMagicinducedComputationalSeparation2024,frauNonstabilizernessEntanglementMatrix2024,
tirritoQuantifyingNonstabilizernessEntanglement2024,
qianQuantumNonLocalMagic2025}.
\par 
The concept of entanglement can be extended in operator space \cite{zanardiEntanglementQuantumEvolutions2001,
wangQuantumEntanglementUnitary2002}, providing a measure of operator complexity across a given bipartition. This idea of operator entanglement was shown to be directly linked with information scrambling and entropy production under unitary evolutions \cite{styliarisInformationScramblingBipartitions2021}, while also proving useful in the context of classical simulability and dynamical complexity of quantum many-body systems \cite{prosenEfficiencyClassicalSimulations2007,
muthDynamicalSimulationIntegrable2011,prosenOperatorSpaceEntanglement2007}. Recently, a Heisenberg picture generalization of nonstabilizerness was introduced in Ref. \cite{dowlingMagicHeisenbergPicture2024}, based on the representation of an operator in the Pauli basis. In addition, in Ref. \cite{dowlingBridgingEntanglementMagic2025} it was shown that this notion of operator space nonstabilizerness provides an upper bound to operator entanglement, making the generation of operator entanglement \cite{andreadakisOperatorSpaceEntangling2024} a sufficient but not necessary condition for the generation of operator space nonstabilizerness.
\par 
In this Letter, we establish that for unitary evolutions, the generation of operator entanglement on Pauli strings is both a necessary and a sufficient condition for the generation of nonlocal nonstabilizerness by the inverse evolution. Nonlocal nonstabilizerness is a form of nonstabilizerness of quantum states that cannot be erased via local operations \cite{caoGravitationalBackreactionMagical2024,caoNontrivialAreaOperators2024,odavicComplexityFrustrationNew2023,catalanoMagicPhaseTransition2024}. To our knowledge, this provides the first exact correspondence between entanglement and nonstabilizerness generation in quantum systems \footnote{Notably, Ref. \cite {caoGravitationalBackreactionMagical2024} provides a relation between nonlocal nonstabilizerness and the spectrum anti-flatness of quantum states, which is a quantum property that depends on both the nonstabilizerness and entanglement.}. We emphasize that our result links an entangling property in \emph{operator space} and nonlocal nonstabilizerness generation in \emph{state space}.
\par Our main result can be summarized as follows:
\par \emph{For a unitary evolution of a system of qubits, all Heisenberg evolved Pauli strings are unentangled over a given bipartition if and only if the inverse Schr\"{o}dinger evolution is a Clifford unitary up to post-processing with local unitaries.}
\par We note that the evolution map in the Heisenberg picture is the inverse of the one in the Schr\"{o}dinger picture, which is the reason why an operator space property of the forward evolution is connected to a state space property of the backwards evolution. The formal statement of the above result is given in \cref{theorem}, while the proof can be found in Ref. \cite{supp}. Guided by the above correspondence, we introduce the Pauli-entangling power, which is defined as the average operator entanglement generated by a unitary evolution when acting on Pauli strings. We compute this average analytically and obtain upper bounds that can be expressed in terms of nonstabilizerness generation in operator space. In the process, we elucidate how operator space nonstabilizerness can be thought of as a form of coherence in operator space, which immediately allows us to formulate a slightly stronger version of a result in Ref. \cite{dowlingBridgingEntanglementMagic2025}. Calculations of the typical value for the Pauli-entangling power and numerical simulations of local Hamiltonian models show that non-local nonstabilizerness generation is ubiquitous and that for long time-scales its behavior under integrability breaking is directly related to that of scrambling.
\par \prlsection{Operator entanglement and nonlocal nonstabilizerness}
Our quantum system of interest is a collection of $N$ qubits $\mathcal{H} \cong (\mathbb{C}^2)^{\otimes N}$ alongside a bipartition $\mathcal{H} \cong \mathcal{H}_A \otimes \mathcal{H}_B \cong (\mathbb{C}^2)^{\otimes N_A} \otimes (\mathbb{C}^2)^{\otimes N_B}$ and define $d=2^N$, $d_A=2^{N_A}$, $d_B=2^{N_B}$. We denote as $\mathscr{U}_N$, $\mathscr{P}_N$ , $\mathscr{C}_N$ the unitary, Pauli, and Clifford groups on $N$ qubits and define as $\tilde{\mathscr{P}}_N = \mathscr{P}_N / \mathbb{Z}_4$ the Abelian group of Pauli strings up to phases. We use the notation $\mathlarger{\mathbb{E}}_{g \in G}$ to denote the Haar average over a compact group $G$ \footnote{If $G$ is discrete, then this is simply the uniform average over all elements of $G$.}. Before stating our main result (\cref{theorem}), we briefly review the relevant concepts of operator entanglement and nonlocal nonstabilizerness.
\par Given a unitary operator $O \in \mathscr{U}_N$, we can always express it in the form $O/\sqrt{d}= \sum_i \sqrt{\lambda_i} V_i \otimes W_i$, where $\{V_i\}_i$ and $\{W_i\}_i$ are orthonormal sets with respect to the Hilbert-Schmidt scalar product and $\{\lambda_i \in \mathbb{R}^+\}_i$ forms a probability distribution, $\sum_i \lambda_i = 1$. Then, we can define the operator entanglement of $O$ as an entropic measure over this probability distribution, e.g., $E_\alpha (O) = \frac{1}{1-\alpha} \log \sum_i \lambda_i^\alpha$ is the operator $\alpha$-R\'{e}nyi entropy. An analytically tractable measure of operator entanglement is given by the linear entropy $E_{\text{lin}}(O)=1-\sum_i \lambda_i^2$, which is also given by the formula \cite{zanardiEntanglementQuantumEvolutions2001}
\begin{equation} \label{opent}
E_{\text{lin}}(O)= 1- \frac{1}{d^2} \Tr( T^A_{12} \, O^{\otimes 2} \, T_{12}^A{O^\dagger}^{\otimes 2}),
\end{equation}
where $T_{12}^A$ is the partial swap between the $A$ subsystems of the doubled space $\mathcal{H}^1 \otimes \mathcal{H}^2 \cong \mathcal{H}^1_A \otimes \mathcal{H}^1_B \otimes \mathcal{H}^2_A \otimes \mathcal{H}^2_B$. 
\par Given a quantum state $\ket{\psi}$, the set $\{d^{-1} \mel{\psi}{P}{\psi}^2 \, \vert \, P \in \tilde{\mathscr{P}_N} \}$ forms a probability distribution over the Pauli strings. Then, the stabilizer R\'{e}nyi entropies are defined as entropic measures over this probability distribution up to a constant shift \cite{leoneStabilizerRenyiEntropy2022}
\begin{equation}
m_\alpha(\ket{\psi}) = \frac{1}{1-\alpha} \log \sum_{P \in \tilde{\mathscr{P}}_N} \frac{\mel{\psi}{P}{\psi}^{2 \alpha}}{d^\alpha} - \log d.
\end{equation}
The strictly nonlocal part of this entropy with respect to the bipartition $\mathcal{H} \cong \mathcal{H}_A \otimes \mathcal{H}_B$ may be defined as \cite{caoGravitationalBackreactionMagical2024,caoNontrivialAreaOperators2024}
\begin{equation}
m_\alpha^{\text{NL}} (\ket{\psi}) = \min_{\substack{U_A \in \mathscr{U}_{N_A} \\ U_B \in \mathscr{U}_{N_B}}} m_\alpha(U_A \otimes U_B \ket{\psi}).
\end{equation}
\par When we write $E(O)$ ($m(\ket{\psi}$), we mean that the statement holds for \emph{any} choice of entropic measure for the operator entanglement (nonstabilizerness).
\par We are now in the position to state formally our main result.
\begin{thm}[\cite{supp}] \label{theorem}
Let $U \in \mathscr{U}_N$ and $\mathcal{H} \cong \mathcal{H}_A \otimes \mathcal{H}_B \cong (\mathbb{C}^2)^{N_A} \otimes (\mathbb{C}^2)^{\otimes N_B}$ be a bipartition of an $N$-qubit quantum system, then
\begin{equation*}
\begin{split}
&E(U^\dagger PU) = 0 \; \forall \, P \in \tilde{\mathscr{P}}_N \Leftrightarrow \\ &\Leftrightarrow \exists \, V \in \mathscr{U}_{N_A}, \, W \in \mathscr{U}_{N_B}, \, C \in \mathscr{C}_N \text{ s.t. } U^\dagger=(V\otimes W) \, C
\end{split}
\end{equation*}
\end{thm}
\begin{proof}[Proof (sketch)]
The $\Leftarrow$ direction follows immediately from the fact that Clifford unitaries map Pauli strings to Pauli strings (up to a sign) and that operator entanglement is invariant under local unitaries, hence $E(V \otimes W C P C^\dagger V^\dagger \otimes W)= E(P^\prime ) =0$, where we also used the fact that Pauli strings are product operators over all qubits.
\par For the $\Rightarrow$ direction the starting point is noticing that $E(U^\dagger PU) = 0 \; \forall \, P \in \tilde{\mathscr{P}}_N$ implies that there are operators $X_P, Y_P$ such that $U^\dagger P U = X_P \otimes Y_P$. Starting from this, we can show that $\{X_P\}$ and $\{Y_P\}$ form appropriate projective irreducible representations of $\tilde{\mathscr{P}}_N$ and using their properties deduce that they have to be unitarily equivalent to Pauli strings (up to signs) $P^\prime_A$ acting on $\mathcal{H}_{N_A}$ and $P^\prime_B$ acting on $\mathcal{H}_{N_B}$, respectively, i.e. $U^\dagger P U = \pm \, V\otimes W \, P_A^\prime \otimes P_B^\prime \, V^\dagger \otimes W^\dagger$. The Clifford unitary $C$ is then just the appropriate mapping between the original and final pauli strings $P$ and $\pm P^\prime_A \otimes P^\prime_B$. Since the above mapping holds for all Pauli strings, which form an operator basis, it follows that it holds for all operators, proving the statement.
\end{proof}
When a unitary operator has the form $U = (V \otimes W) \, C$ for some $V \in \mathscr{U}_{N_A}, \, W \in \mathscr{U}_{N_B}, \, C \in \mathscr{C}_N$, we say that $U$ generates \emph{no} nonlocal nonstabilizerness. The reasoning for this definition becomes clear by the following Lemma.
\begin{lemma}[\cite{supp}] \label{lemma_no_NL}
For all unitaries $U \in \mathscr{U}_{N}$ s.t. $U= (V \otimes W) \, C$, where $V \in \mathscr{U}_{N_A}, \, W \in \mathscr{U}_{N_B}, \, C \in \mathscr{C}_N$,
\begin{equation} \label{no_NL}
m^{\text{NL}}(U \ket{\psi}) \leq m(\ket{\psi}).
\end{equation}
In particular, if $\ket{s}$ is a stabilizer state, then $m^{\text{NL}}(U \ket{s})=0$.
\end{lemma}
\cref{no_NL} expresses the intuitive idea that any unitary of the form $U = (V \otimes W) \, C$ can at most spread pre-existing nonstabilizerness across the bipartition--due to the potentially global action of the Clifford unitary--but cannot generate additional nonlocal nonstabilizerness. From this perspective, \cref{theorem} captures a duality in the growth of quantum complexity: The generation of operator entanglement via Heisenberg evolution on Pauli strings is intrinsically linked to the generation of nonlocal nonstabilizerness under the inverse Schr\"{o}dinger evolution, such that one cannot occur without the other. Let us emphasize once more how unique this result is. When looking strictly at properties of quantum states, entanglement generation is \emph{not} sufficient for nonlocal nonstabilizerness generation. For example, any entangling Clifford unitary, such as the CNOT in a 2-qubit system, will generate entanglement but will not generate nonlocal nonstabilizerness--in fact, it does not produce any nonstabilizerness.

\prlsection{Average Pauli-entangling power}
\begin{figure*}
\centering
\begin{subfigure}{0.35\textwidth}
\centering
\begin{tikzpicture}
\def\ds{1.05}
\node at (-1.4*\ds,0) {$\text{\large $U=$}$};
\foreach \i in {0,1,3}{
\def\x{\i*\ds}
\Agate{\x}{0}
\ifnum\i =1
\draw[thick] (\x+0.35,0) --(\x+0.7,0);
\draw[thick] (\x-0.35,0) -- (\x-0.7,0);
\else
\ifnum\i =0
\draw[thick] (\x+0.35,0) --(\x+0.7,0);
\else
\draw[thick] (\x-0.35,0) -- (\x-0.7,0);
\fi
\fi
}
\node at (2*\ds,0) {$\dots$};
\TrLeft{(-0.35,0)}
\TrRight{(3*\ds+0.35,0)}
\end{tikzpicture}
\caption{}
\end{subfigure}%
\hfill
\begin{subfigure}{0.6\textwidth}
\centering
\begin{tikzpicture}
\def\ds{1.05}
\node at (-0.4*\ds,0) {$\text{\large $\Lambda=$}$};
\foreach \i in {1,2,3,4}{
\def\x{\i*\ds}
\Pigate{\x}{0}
\ifnum\i=1
\draw[thick] (\x+0.35,0) --(\x+0.7,0);
\else \ifnum\i=4
\draw[thick] (\x-0.35,0) -- (\x-0.7,0);
\else
\draw[thick] (\x+0.35,0) --(\x+0.7,0);
\draw[thick] (\x-0.35,0) -- (\x-0.7,0);
\fi
\fi}
\TrLeft{(\ds-0.35,0)}
\TrRight{(4*\ds+0.35,0)}
\node at (4*\ds+1.4*\ds,0) {$,$ where };
\Pigate{5.4*\ds+1.4*\ds}{0}
\draw[thick] (6.8*\ds+0.35,0) --(6.8*\ds+0.7,0);
\draw[thick] (6.8*\ds-0.35,0) -- (6.8*\ds-0.7,0);
\node at (6.8*\ds-0.7, 0.2) {$\scriptstyle\alpha$};
\node at (6.8*\ds+0.7, 0.2) {$\scriptstyle\beta$};
\node at (6.8*\ds+0.7, 0) {$\qquad \qquad = \delta_{\alpha \beta}$};
\sigmagate{6.8*\ds+2.2}{0}
\end{tikzpicture}
\caption{}
\end{subfigure}%
\vspace{20pt}
\begin{subfigure}{\textwidth}
\centering
\begin{tikzpicture}
\def\ds{1.05}
\def\dy{0.25}
\node at (-1.6,0) {$\text{\large $P_E(U)=1-d^{-4}$ \Huge Tr}$};
\begin{scope}[scale=0.8]
\foreach \i in {0,1}{
\def\h{\i*\dy}
\def\x{1.8*\i*\ds}
\Agate{2+\x}{\ds+\h}
\Pigate{2+\x}{0+\h}
\Adaggate{2+\x}{-\ds+\h}
\draw[thick] (2+\x-0.35,\ds+\h) -- (2+\x-0.7, \ds+\h);
\draw[thick] (2+\x+0.35, \ds+\h) -- (2+\x+0.7, \ds+\h);
\draw[thick] (2+\x-0.35,-\ds+\h) -- (2+\x-0.7, -\ds+\h);
\draw[thick] (2+\x+0.35, -\ds+\h) -- (2+\x+0.7, -\ds+\h);
\begin{pgfonlayer}{background}
\draw[thick] (2,0) -- (2+1.8*\ds,\dy);
\end{pgfonlayer}
}
\draw[thick] (2,\ds+0.6) .. controls (2,\ds+0.6+0.6) and (2+0.8,\ds++0.6+1) .. (2+0.9*\ds,\dy/2)
.. controls (2*2+2*0.9*\ds-2-0.8,2*\dy/2-\ds-0.6-1) and (2+1.8*\ds,-\ds+\dy-0.6-1) .. (2+1.8*\ds,-\ds+\dy-0.6);
\draw[thick] (2,-\ds-0.6) .. controls (2,-\ds-0.6-0.6) and (2+0.8,-\ds-0.6-1) .. (2+0.9*\ds,\dy/2)
.. controls (2*2+2*0.9*\ds-2-0.8,2*\dy/2+\ds+0.6+1) and (2+1.8*\ds,\ds+\dy+0.6+1) .. (2+1.8*\ds,\ds+\dy+0.6);
\def\deltax{3.6*\ds}
\def\deltay{2*\dy}
\begin{pgfonlayer}{background}
\draw[thick] (2+1.8*\ds,\dy) -- (2+\deltax,\deltay);
\end{pgfonlayer}
\begin{scope}[shift={(\deltax,\deltay)}]
\foreach \i in {0,1}{
\def\h{\i*\dy}
\def\x{1.8*\i*\ds}
\Agate{2+\x}{\ds+\h}
\Pigate{2+\x}{0+\h}
\Adaggate{2+\x}{-\ds+\h}
\draw[thick] (2+\x-0.35,\ds+\h) -- (2+\x-0.7, \ds+\h);
\draw[thick] (2+\x+0.35, \ds+\h) -- (2+\x+0.7, \ds+\h);
\draw[thick] (2+\x-0.35,-\ds+\h) -- (2+\x-0.7, -\ds+\h);
\draw[thick] (2+\x+0.35, -\ds+\h) -- (2+\x+0.7, -\ds+\h);
\begin{pgfonlayer}{background}
\draw[thick] (2,0) -- (2+1.8*\ds,\dy);
\end{pgfonlayer}
}
\draw[thick] (2,\ds+0.6) .. controls (2,\ds+0.6+0.6) and (2+0.8,\ds++0.6+1) .. (2+0.9*\ds,\dy/2)
.. controls (2*2+2*0.9*\ds-2-0.8,2*\dy/2-\ds-0.6-1) and (2+1.8*\ds,-\ds+\dy-0.6-1) .. (2+1.8*\ds,-\ds+\dy-0.6);
\draw[thick] (2,-\ds-0.6) .. controls (2,-\ds-0.6-0.6) and (2+0.8,-\ds-0.6-1) .. (2+0.9*\ds,\dy/2)
.. controls (2*2+2*0.9*\ds-2-0.8,2*\dy/2+\ds+0.6+1) and (2+1.8*\ds,\ds+\dy+0.6+1) .. (2+1.8*\ds,\ds+\dy+0.6);
\end{scope}
\begin{pgfonlayer}{background}
\begin{scope}[shift={(2,0)}]
\draw [thick] (0,0) to (-0.7,-0.7*0.25/1.89);
\draw [thick] (-0.7,-0.7*0.25/1.89) to  [bend left=90] (-0.7,-0.7*0.25/1.89+\ds/2);
\draw [thick, dotted] (-0.7,-0.7*0.25/1.89+\ds/2) to  (-0.25,-0.25*0.25/1.89+\ds/2);
\end{scope}
\begin{scope}[shift={(2+3*1.8*\ds,3*\dy)},xscale=-1,rotate=-15.04]]
\draw [thick] (0,0) to (-0.7,-0.7*0.25/1.89);
\draw [thick] (-0.7,-0.7*0.25/1.89) to  [bend left=90] (-0.7,-0.7*0.25/1.89+\ds/2);
\draw [thick, dotted] (-0.7,-0.7*0.25/1.89+\ds/2) to  (-0.25,-0.25*0.25/1.89+\ds/2);
\end{scope}
\end{pgfonlayer}
\draw[thick] (2.2-0.3-0.2,3.6+2*\dy) -- (2.2-0.3-1.2,3.6+2*\dy) -- (2.2-0.3-1.2,-3.6+2*\dy) -- (2.2-0.3-0.2,-3.6+2*\dy);
\draw[thick,decorate,decoration={bent,aspect=0.3,amplitude=-30pt}] (2.2-0.3,3+2*\dy) -- (2.2-0.3,-3+2*\dy);
\draw[thick,decorate,decoration={bent,aspect=0.3,amplitude=30pt}] (1.8+3*1.8*\ds+0.3,3+2*\dy) -- (1.6+3*1.8*\ds+0.3,-3+2*\dy);
\node at (1.8+3*1.8*\ds+0.3+0.5,3+2*\dy) {$N_A$};
\begin{scope}[shift={(8.1,0)}]
\foreach \i in {0,1,2,3}{
\begin{scope}[shift={(0.2+\i*1.7*\ds,\i*\dy)}]
\Agate{2}{\ds}
\Pigate{2}{0}
\Adaggate{2}{-\ds}
\draw[thick] (2-0.35,\ds) -- (2-0.7, \ds);
\draw[thick] (2+0.35, \ds) -- (2+0.7, \ds);
\draw[thick] (2-0.35,-\ds) -- (2-0.7, -\ds);
\draw[thick] (2+0.35, -\ds) -- (2+0.7, -\ds);
\draw[thick] (2,\ds+0.6) .. controls (2,\ds+0.6+0.6) and (2-0.85,\ds++0.6+1) .. (2-0.85,\dy/2)
.. controls (2*2-2*0.85-2+0.8,2*\dy/2-\ds-0.6-1) and (2,-\ds-0.6-1) .. (2,-\ds-0.6);
\end{scope}}
\foreach \i in {0,1,2}{
\begin{scope}[shift={(0.2+\i*1.7*\ds,\i*\dy)}]
\begin{pgfonlayer}{background}
\draw[thick] (2,0) -- (2+1.7*\ds,\dy);
\end{pgfonlayer}
\end{scope}}
\begin{pgfonlayer}{background}
\begin{scope}[shift={(2+0.2,0)}]
\draw [thick] (0,0) to (-0.6,-0.6*0.25/1.785);
\draw [thick] (-0.6,-0.6*0.25/1.785) to  [bend left=90] (-0.6,-0.6*0.25/1.785+\ds/2);
\draw [thick, dotted] (-0.6,-0.6*0.25/1.785+\ds/2) to  (-0.2,-0.2*0.25/1.785+\ds/2);
\end{scope}
\begin{scope}[shift={(2+0.2+3*1.7*\ds,3*\dy)},xscale=-1,rotate=-15.94]]
\draw [thick] (0,0) to (-0.6,-0.6*0.25/1.785);
\draw [thick] (-0.6,-0.6*0.25/1.785) to  [bend left=90] (-0.6,-0.6*0.25/1.785+\ds/2);
\draw [thick, dotted] (-0.6,-0.6*0.25/1.785+\ds/2) to  (-0.2,-0.2*0.25/1.785+\ds/2);
\end{scope}
\end{pgfonlayer}
\draw[thick] (1.3+3*1.8*\ds+1,3.6+2*\dy) -- (1.3+3*1.8*\ds+2,3.6+2*\dy) -- (1.3+3*1.8*\ds+2,-3.6+2*\dy) -- (1.3+3*1.8*\ds+1,-3.6+2*\dy);
\draw[thick,decorate,decoration={bent,aspect=0.3,amplitude=-30pt}] (2-0.3,3+2*\dy) -- (2-0.3,-3+2*\dy);
\draw[thick,decorate,decoration={bent,aspect=0.3,amplitude=30pt}] (1.7+3*1.8*\ds+0.3,3+2*\dy) -- (1.7+3*1.8*\ds+0.3,-3+2*\dy);
\node at (1.7+3*1.8*\ds+0.3+0.5,3+2*\dy) {$N_B$};
\end{scope}
\end{scope}
\end{tikzpicture}
\caption{}
\label{fig:formula}
\end{subfigure}%
\caption{(a) A uniform matrix product unitary $U$ is represented by a single tensor $A$, with the bond dimension $\chi$ being the dimension of the inner (contracted) legs. (b) The $\Lambda$ operator of \cref{q} can be represented as a uniform matrix product unitary using the tensor $\Pi$, which is equal to a Pauli operator when both inner legs have the same index and zero otherwise. (c) The Pauli-entangling power \cref{paul_ent} can be written in terms of powers of two transfer matrices, which contain four copies of $A$ and four copies of $A^\dagger$, hence a total bond dimension equal to $\chi^8$.}
\label{fig:mpu}
\end{figure*}
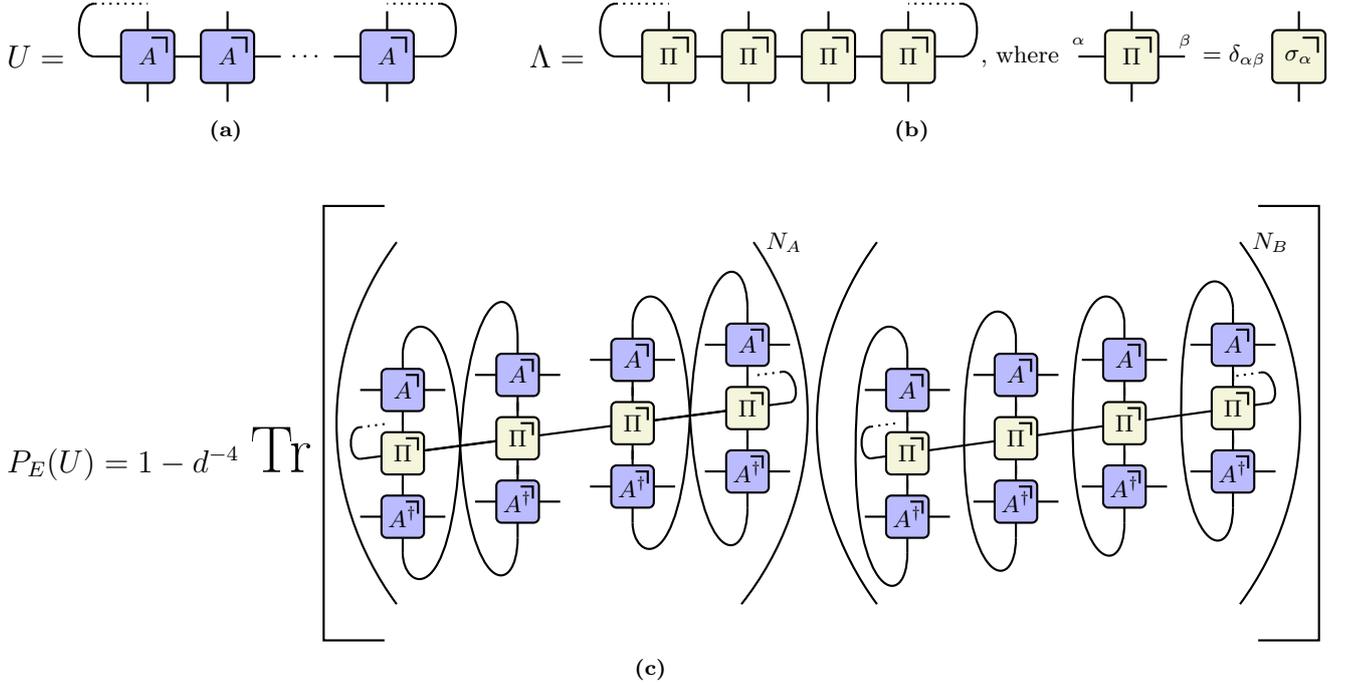
\cref{theorem} provides a novel operational meaning on the growth of operator entanglement for Pauli strings evolved in the Heisenberg picture. Here, we will provide an average measure of this operator entanglement increase given some unitary dynamics $U$. We will quantify operator entanglement by means of the linear entropy, see \cref{opent}, which allows us to perform the average over the Pauli strings analytically. Recall that the linear entropy is related to the R\'{e}nyi entropies via the relation
\begin{equation}
E_\alpha(O) \geq E_2 (O) = - \log(1-E_{\text{lin}}(O)), \; \alpha \in [0,2).
\end{equation}
Additionally, when computing averages over a distribution of $O$, Jensen's inequality transforms the last equality into a lower bound for $\mathlarger{\mathbb{E}}_{O} E_2(O)$ in terms of $\mathlarger{\mathbb{E}}_O E_{\text{lin}}(O)$.
\begin{proposition}[\cite{supp}]\label{prop_paul_ent}
Let $U \in \mathscr{U}_N$ and $\mathcal{H} \cong \mathcal{H}_A \otimes \mathcal{H}_B \cong (\mathbb{C}^2)^{N_A} \otimes (\mathbb{C}^2)^{\otimes N_B}$, then the average Pauli-entangling power is given as
\begin{equation} \label{paul_ent}
\begin{split}
P_E(U) &\coloneqq \mathlarger{\mathbb{E}}_{P \in \tilde{\mathscr{P}}_N} E_{\text{lin}}(U^\dagger P U) \\
&= 1- \frac{1}{d^2} \Tr( T^A_{(12)(34)} \, {U^\dagger}^{\otimes 4} Q U^{\otimes 4}),
\end{split}
\end{equation}
where $Q \coloneqq d^{-2} \sum_{P \in \tilde{\mathscr{P}_N}} P^{\otimes 4}$.
\end{proposition}
Here, $T^A_{(12)(34)}$ is a partial permutation between the copies of the $A$ subsystem in the quadrupled space $\mathcal{H}^{\otimes 4}$. $Q$ in \cref{paul_ent} is a rank-$d^2$ projector that plays a fundamental role in the representation theory of the Clifford group \cite{zhuCliffordGroupFails2016}. In addition, its image ${U^\dagger}^{\otimes 4} Q U^{\otimes 4}$ under the adjoint action of ${U^\dagger}^{\otimes 4}$ is directly related to the nonstabilizing power of $U^\dagger$ \cite{leoneStabilizerRenyiEntropy2022}.
\par The following invariance property of $P_E(U)$ follows directly from the definition \cref{paul_ent}:
\begin{equation}
P_E(C \, U \,V_A \otimes V_B)=P_E(U),
\end{equation}
$\forall \, U\in \mathscr{U}_N, V_A \in \mathscr{U}_{N_A} V_B \in \mathscr{U}_{N_B}, C\in \mathscr{C}_N$. This expresses the fact that the Pauli-entangling power is invariant under post-processing with a Clifford unitary and pre-processing with local unitaries. This is an intuitive property for a proxy of nonlocal nonstabilizerness generation for the inverse evolution $V_A^\dagger \otimes V_B^\dagger \,  U^\dagger \, C^\dagger$: Clifford unitaries preserve the initial amount of nonstabilizerness, while local unitaries can, at most, produce nonstabilizerness locally.
\par In general, the expression in \cref{paul_ent} is hard to compute for generic evolutions and large system sizes. However, we can adapt the method of Ref. \cite{haugQuantifyingNonstabilizernessMatrix2023} to obtain an efficient computation for unitary evolutions that can be efficiently represented as matrix product unitaries \cite{ignaciociracMatrixProductUnitaries2017,styliarisMatrixproductUnitariesQuantum2025}. First, notice that $Q$ may be equivalently expressed as 
\begin{equation} \label{q}
Q = d^{-2} \Lambda^{\otimes N},
\end{equation}
where $\Lambda \coloneqq \sum_{\sigma \in \tilde{\mathscr{P}_1}} \sigma^{\otimes 4}$. Notice that now $\sigma$ are single-qubit Pauli operators and the equality in \cref{q} should be understood after the appropriate rearrangement of the tensor factors, $\bigotimes_{j=1}^4 (\mathbb{C}^2)^{\otimes N} \cong \bigotimes_{i=1}^N (\mathbb{C}^2)^{\otimes 4}$. Then, if we assume that $U$ is a matrix product unitary with bond dimension $\chi$, \cref{paul_ent} corresponds to a tensor contraction for a matrix product operator of bond dimension $\chi^8$. For simplicity, in \cref{fig:mpu} we further assume that $U$ is a uniform matrix product unitary, which have been shown to coincide with the class of translation-invariant quantum cellular automata \cite{ignaciociracMatrixProductUnitaries2017,
farrellyReviewQuantumCellular2020}. As seen in \cref{fig:formula}, the Pauli-entangling power is then expressed in terms of powers of two transfer matrices with bond dimension equal to $\chi^8$. In fact, as $N_A,N_B \rightarrow \infty$, $P_E(U)$ depends only on the largest eigenvalues of these matrices--assuming they are unique.
\par \prlsection{Upper bounds \& coherence in operator space}
In order to further investigate the analytical properties of the Pauli-entangling power \cref{paul_ent}, it is desirable to obtain relations to other physical quantities. To this end, the notion of nonstabilizerness generation in operator space proves useful. This was recently quantified in Ref. \cite{dowlingMagicHeisenbergPicture2024} via the so-called operator stabilizer entropies. Specializing on initial Pauli operators $P \in \tilde{\mathscr{P}}_N$, these read as
\begin{equation}
M_\alpha (U^\dagger P U) \coloneqq \frac{1}{1-\alpha} \log \sum_{P^\prime \in \tilde{\mathscr{P}}_N} \left(\frac{\Tr(U^\dagger P U P^\prime)}{d} \right)^{2 \alpha}.
\end{equation}
The corresponding linear operator stabilizer entropy is then given as $M_{\text{lin}}(U^\dagger P U) = 1- d^{-4} \sum_{P^\prime \in \tilde{\mathscr{P}}_N} \Tr(U^\dagger P U P^\prime)^{4}$. Here, we would like to point out that this linear operator stabilizer entropy can be understood simply as a 2-coherence in operator space. Specifically, recall that given a pure quantum state $\rho=\ketbra{\psi}{\psi}$ and a basis $B=\{\Pi_i=\ketbra{i}{i} \}_{i=1}^d$, its 2-coherence is given as $c_{2,B}=\lVert \rho - \sum_{i=1}^d \Pi_i \rho \Pi_i \rVert_2^2 = 1 - \sum_{i=1}^d \lvert \braket{i}{\psi} \rvert^4$ \cite{anandQuantumCoherenceSignature2021}. Translating this to operator space, where the role of states is played by unitaries and $B$ is an operator basis, we simply observe that:
\begin{obs} \label{observation}
For any $U\in \mathscr{U}_N$,
\begin{equation}
M_{\text{lin}}(U) =c_{2,B}\left( \frac{U}{\sqrt{d}}\right) ,
\end{equation}
where $B=\{ P^\prime / \sqrt{d} \, \vert \, P^\prime \in \tilde{\mathscr{P}}_N \}$ and $c_{2,B}(U)= 1 - \sum_{P^\prime \in \tilde{\mathscr{P}}_N} \left\vert \Tr(U/\sqrt{d} \, P^\prime/\sqrt{d}) \right\vert^4$.
\end{obs}
While this is a simple observation, it allows us to use results from the resource theory of coherence \cite{anandQuantumCoherenceSignature2021} adapted in operator space. In particular, this allows for a direct relation between entanglement and nonstabilizerness in operator space:
\begin{lemma}[\cite{supp}] \label{lemma_opcoh}
For any unitary operator $U \in \mathscr{U}_N$,
\begin{equation} \label{opcoh}
E_{\text{lin}}(U)= \min_{{\substack{V_A, W_A \in \mathscr{U}_{N_A} \\ V_B, W_B \in \mathscr{U}_{N_B}}}} M_{\text{lin}}(V_A \otimes V_B \, U \, W_A \otimes W_B).
\end{equation}
\end{lemma}
\cref{opcoh} is directly related to the fact that the operator stabilizer entropies upper bound operator entanglement, see Theorem 2 in Ref. \cite{dowlingMagicHeisenbergPicture2024}. \cref{opcoh} also allows us to express the Pauli-entangling power \cref{paul_ent} as the average operator nonstabilizerness of the evolved Pauli strings $U^\dagger P U$ up to local unitaries $V_A, W_A \in \mathscr{U}_{N_A}$ and $V_B, W_B \in \mathscr{U}_{N_B}$. Importantly, the local unitaries involved in the minimization will in general depend on $P$ and will in general not constitute unitary channels, $V_{A,B} \neq W_{A,B}^\dagger$. In contrast, \cref{theorem} ensures that when $P_E(U)=0$, we can indeed find local unitary channels, independent of $P$, that erase any operator nonstabilizerness generation:
\begin{cor}[\cite{supp}] \label{cor_no_op_NL}
For any $U \in \mathscr{U}_N$,
\begin{equation} \label{no_op_NL}
\begin{split}
P_E(U) = 0 \Leftrightarrow \, &\exists \, V \in \mathscr{U}_{N_A}, \, W \in \mathscr{U}_{N_B} \text{ s.t. } \forall P \in \tilde{\mathscr{P}}_N: \\ 
&M_{\text{lin}}(V^\dagger \otimes W^\dagger \, U^\dagger PU \, V \otimes W) =0.
\end{split}
\end{equation}
\end{cor}
\cref{no_op_NL} provides a version of \cref{theorem} strictly in operator space, where the generation of operator entanglement on Pauli strings is a necessary and sufficient condition for the generation of operator nonstabilizerness that is not locally erasable.
\par In addition to the above, we can use operator stabilizer entropies to upper-bound the Pauli-entangling power $P_E(U)$. 
\begin{proposition}[\cite{supp}] \label{prop_upbound}
For any $U \in \mathscr{U}_N$,
\begin{equation} \label{upbound}
  \begin{split}
P_E(U) \leq &\min\left\{ {\mathlarger{\mathbb{E}}}_{P_A \in \tilde{\mathscr{P}}_{N_A}} M_{\text{lin}}(U \, P_A \otimes \mathds{1}_B \, U^\dagger),\right. \\
&\hspace{24pt}\left. {\mathlarger{\mathbb{E}}}_{P_B \in \tilde{\mathscr{P}}_{N_B}} M_{\text{lin}}(U \, \mathds{1}_A \otimes P_B \, U^\dagger) \right\},
  \end{split}
\end{equation}
\end{proposition}
namely $P_E(U)$ is bounded by the average operator nonstabilizerness generated by the inverse Heisenberg evolution on \emph{subsystem local} Pauli strings.
This gives a quantitative relation, in the form of an inequality, between the entangling properties of the forward evolution and the nonstabilizerness generating properties of the backwards evolution in operator space. Quite interestingly, although \cref{paul_ent} accounts for operator entanglement generation on all Pauli strings, the bound \cref{upbound} involves only local operators with respect to the system bipartition.
\par\prlsection{Typical value \& quantum scrambling}
Let us turn our attention to the Pauli-entangling properties of random unitaries. Due to measure concentration for randomly distributed unitaries in large dimensions \cite{meckesRandomMatrixTheory2019}, the typical value of the Pauli-entangling power corresponds to the average value over Haar random unitaries. Using \cref{paul_ent} and performing the average over the Haar measure, we obtain that:
\begin{proposition}[\cite{supp}] \label{prop_typical}
  \begin{equation} \label{typical}
    \begin{split}
  &{\mathlarger{\mathbb{E}}}_{U \in \mathscr{U}_N} P_E(U)=\frac{(d^2-d_A^2)(d^2-10)(d_A^2-1)}{d^2 d_A^2(d^2-9)}\\
  &= 1-\left(1-\frac{1}{d^2}\right)\frac{1}{d_A^2}-\left(1-\frac{1}{d^2}\right)\frac{1}{d_B^2} + O\left(\frac{1}{d^4}\right).
    \end{split}
  \end{equation}
  \end{proposition}
Notice that due to the right invariance of the Haar measure, the initial average over Pauli strings in \cref{paul_ent} may be separated into two terms; the identity operator that is invariant and remains unentangled and the non-identity Pauli strings that obtain the same typical value for the operator entanglement. This second contribution was previously computed by taking an appropriate limit in Ref. \cite{dowlingBridgingEntanglementMagic2025}, see Eq. (B50), and using an approximation in Ref. \cite{kudler-flamEntanglementLocalOperators2021}, see Eq. (29), both of which are compatible with \cref{typical}. Notice that for $d_A=d_B=\sqrt{d}\gg 1$, ${\mathlarger{\mathbb{E}}}_{U \in \mathscr{U}_N} P_E(U) \sim 1-\frac{2}{d}$, which is precisely the system size scaling of the typical value for the scrambling of information across a symmetric bipartition for Haar distributed unitary evolutions \cite{styliarisInformationScramblingBipartitions2021}. Intuitively, for generic unitary evolutions in large dimensions, quantum scrambling, as quantified by out-of-time-order-correlators \cite{swingleUnscramblingPhysicsOutoftimeorder2018, styliarisInformationScramblingBipartitions2021}, is accompanied by the growth of the operator entanglement of Heisenberg evolved operators \cite{dowlingScramblingNecessaryNot2023}. In fact, in \cref{appendix} we present numerical evidence that in the long-time limit, the Pauli-entangling power behaves similarly to the scrambling properties of prototypical quantum many-body models as we tune the strength of an integrability breaking term. This implies that different signatures of quantum complexity in the long-time limit of local Hamiltonian evolutions are often interconnected, exhibitng a similar response to integrability breaking \cite{odavicStabilizerEntropyNonintegrable2025} and the onset of random matrix theory \cite{guhrRandomMatrixTheoriesQuantum1998}, despite such evolutions being atypical.
\par \prlsection{Conclusion}
Entanglement and nonstabilizerness are two fundamental resources in quantum information processing that underpin quantum complexity. The emergence of these non-classical properties under quantum evolution may then be used to characterize the complexity of the dynamics. In this Letter, we proved in \cref{theorem} an exact correspondence between the generation of entanglement in operator space and the generation of nonlocal nonstabilizerness by the inverse evolution in state space. We believe that this paves the way for a deeper understanding of the interplay between the ostensibly distinct origins of nonstabilizerness and entanglement growth in quantum systems.
\par \cref{theorem} provides a novel insight on the significance of the operator entanglement increase of Heisenberg evolved Pauli strings. To quantify this phenomenon, we introduced an analytically tractable coarse measure of this operator entanglement increase, referred to as the Pauli-entangling power. Notice that this may be understood as a measure of operator space entangling power \cite{andreadakisOperatorSpaceEntangling2024} over a restricted set of initial unentangled operators. We showed that this Pauli-entangling power may be efficiently computed for unitary evolutions that can be represented as matrix product unitaries with small bond dimension. 
\par In addition, we described how a recently introduced concept of nonstabilizerness in the Heisenberg picture can be understood as coherence generation in operator space. We used this notion of operator nonstabilizerness to formulate a corollary of \cref{theorem} purely in operator space, as well as to obtain upper bounds for the Pauli-entangling power. It is important to emphasize that both the operator entanglement and operator nonstabilizerness affect the computational cost of numerical simulations of the quantum system \cite{dowlingBridgingEntanglementMagic2025}, which is associated with the increase of its quantum complexity.
\par Analytical computations of the typical value, along with numerical simulations of prototypical quantum many-body models, suggest that, in several cases, Pauli-entangling properties are a primary component of quantum scrambling. In particular, we observed that both the Pauli-entangling power and an average measure of quantum scrambling \cite{styliarisInformationScramblingBipartitions2021} are closely related, both in terms of their typical scaling with system size and their long-time behavior under integrability breaking in Hamiltonian models.
\par A natural direction for future work is to investigate the relationship between operator entanglement and nonlocal nonstabilizerness presented in this Letter in concrete settings. Nonlocal nonstabilizerness and its associated quantum complexity have recently been identified as useful concepts, both for characterizing quantum phase transitions in topologically frustrated spin chains \cite{odavicComplexityFrustrationNew2023,catalanoMagicPhaseTransition2024}, as well as for identifying properties of quantum gravity in holographic models \cite{caoNontrivialAreaOperators2024,caoGravitationalBackreactionMagical2024}. It is therefore intriguing to explore whether its connection to operator entanglement can provide further insights in these domains.
\par \prlsection{Acknowledgments}
The Authors acknowledge valuable discussions with Aliosca Hamma and Emanuel Dallas. FA would like to thank Namit Anand and Pavel Kos for useful comments and discussions. FA acknowledges financial support from a University of Southern California ``Philomela Myronis'' scholarship. PZ acknowledges partial support from the NSF award PHY2310227.

\bibliographystyle{apsrev4-2}
\bibliography{MyLibrary.bib}
\onecolumngrid
\newpage
\appendix
\setcounter{secnumdepth}{2}
\crefalias{section}{appendix}  
\section{Quantum spin-chain models} \label{appendix}
We consider the following Hamiltonian models on a one-dimensional lattice of $N$ qubits with periodic boundary conditions:
\begin{enumerate}
  \item {Heisenberg XYZ model in an external magnetic field,
  \begin{equation}\label{xyz}
  H_{XYZ}=\sum_{i=1}^{N} \left( J_x \sigma_i^x \sigma_{i+1}^x + J_y \sigma_i^y \sigma_{i+1}^y + J_z \sigma_i^z \sigma_{i+1}^z + h \sigma_i^z \right),
  \end{equation}
  where $J_x, J_y, J_z$ are the coupling constants and $h$ is the strength of the external magnetic field.
  }
  \item {Transverse field Ising model (TFIM),
  \begin{equation}\label{tfim}
  H_{TFIM}=-\sum_{i=1}^{N} \left( J \sigma_i^z \sigma_{i+1}^z + h \sigma_i^z + g \sigma_i^x \right),
  \end{equation}
  where $J$ is the coupling constant and $h,g$ are the strengths of a longitudal and transverse magnetic field.
  }
\end{enumerate}
Using exact diagonalization we simulate the dynamics $U_t=\exp{-iHt}$ generated by these Hamiltonian models and numerically compute the long-time average, $\overline{X(t)}^t \coloneqq \lim_{T \rightarrow \infty} 1/T \int_0^T X(t) dt$, of both the Pauli-entangling power \cref{paul_ent} and the operator entanglement $E_{\text{lin}}(U_t)$ of the unitary evolution operator, for small system sizes ($N=8,11$)\, \footnote{The code used for the simulations is publicly available in Ref. \cite{andreadakisGithubNlmagic}.}. For the bipartition we choose the first $\lfloor N/2 \rfloor$ to be subsystem $A$ and the rest to be subsystem $B$. We note that $E_{\text{lin}}(U_t)$ is closely related to out-of-time-order correlation functions, providing an average measure of quantum scrambling, with its system size scaling being sensitive to the integrability properties of local Hamiltonian models \cite{styliarisInformationScramblingBipartitions2021}.
\par \cref{fig:models} compares the behavior of the long-time average of these two quantities for fixed values $J_x=0.75, J_y=0.25, h=0.5$ for the XYZ and $g=1$ for the TFIM model, while varying $J_z$ and $h$ respectively. Notice that for $J_z=0$ ($h=0$) the XYZ (TFIM) model can be mapped to free fermions \cite{liebTwoSolubleModels1961} and is nonintegrable otherwise \cite{yamaguchiCompleteClassificationIntegrability2024}. We observe that the long-time values of both $P_E(U_t)$ and $E_{\text{lin}}(U_t)$ behave similarly as we tune up the integrability breaking term for a given system size. This indicates that in various prototypical Hamiltonian models, the increased long-time quantum scrambling associated with nonintegrable evolutions is accompanied by an increased operator complexity of the Heisenberg evolved Pauli strings.
\par The parameters of the numerical simulations are as follows: we have set empirically $dt=0.2$ and we evolve the system until the standard error of the mean for the long-time average is below a certain threshold, $1.96\, \sigma/\sqrt{N_t}<2\cdot 10^{-2}$, where $\sigma$ is the sample standard deviation and $N_t$ is the number of timesteps. For $N=8$, we evaluate $P_E(U_t)$ and $E(U_t)$ exactly using \cref{opent,paul_ent}, while for $N=11$ we compute $P_E(U_t)$ by sampling over the Pauli strings until the standard error of the mean for each timestep is below the same threshold as before.
\begin{figure}[h]
  \centering
  \begin{subfigure}{0.58\textwidth}
    \centering
  \includegraphics[width=1\textwidth]{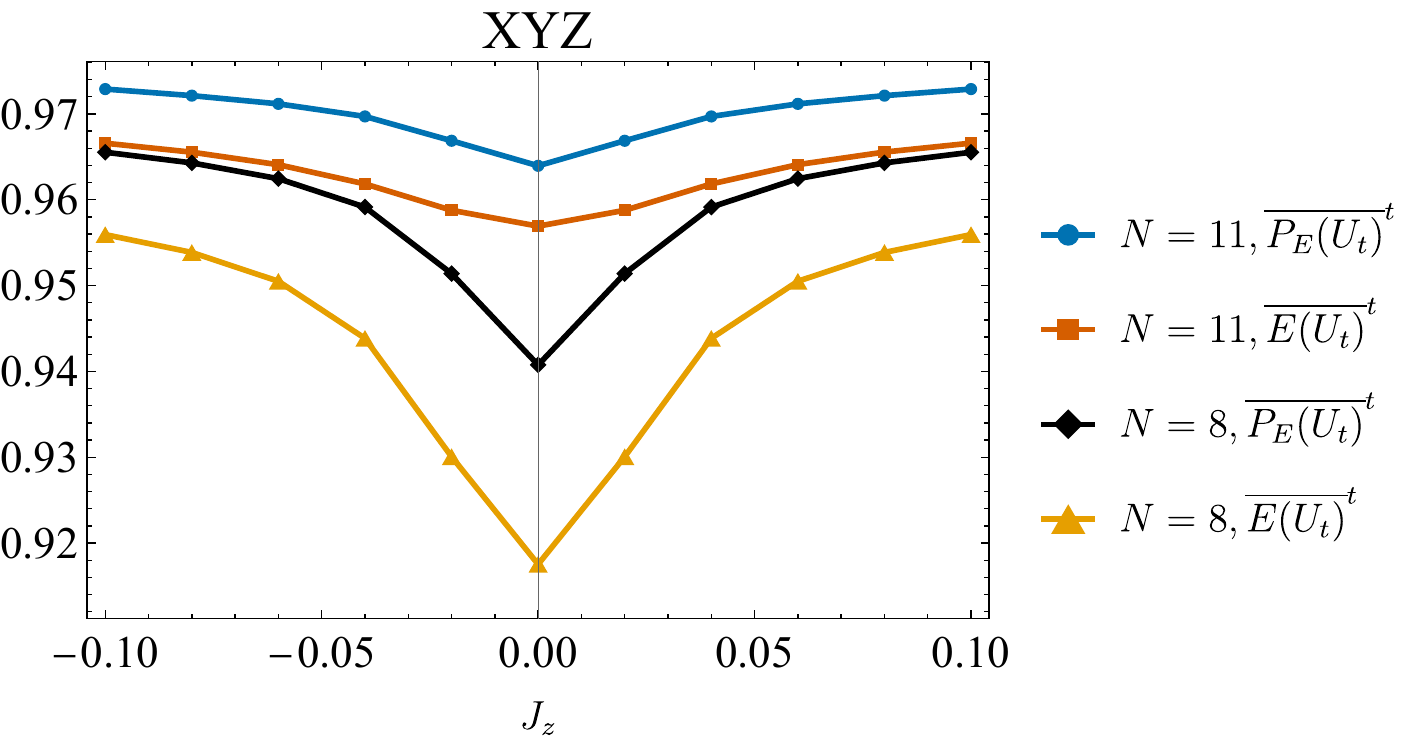}
  \caption{}
  \end{subfigure}%
  \begin{subfigure}{0.42\textwidth}
    \centering
  \includegraphics[width=1\textwidth]{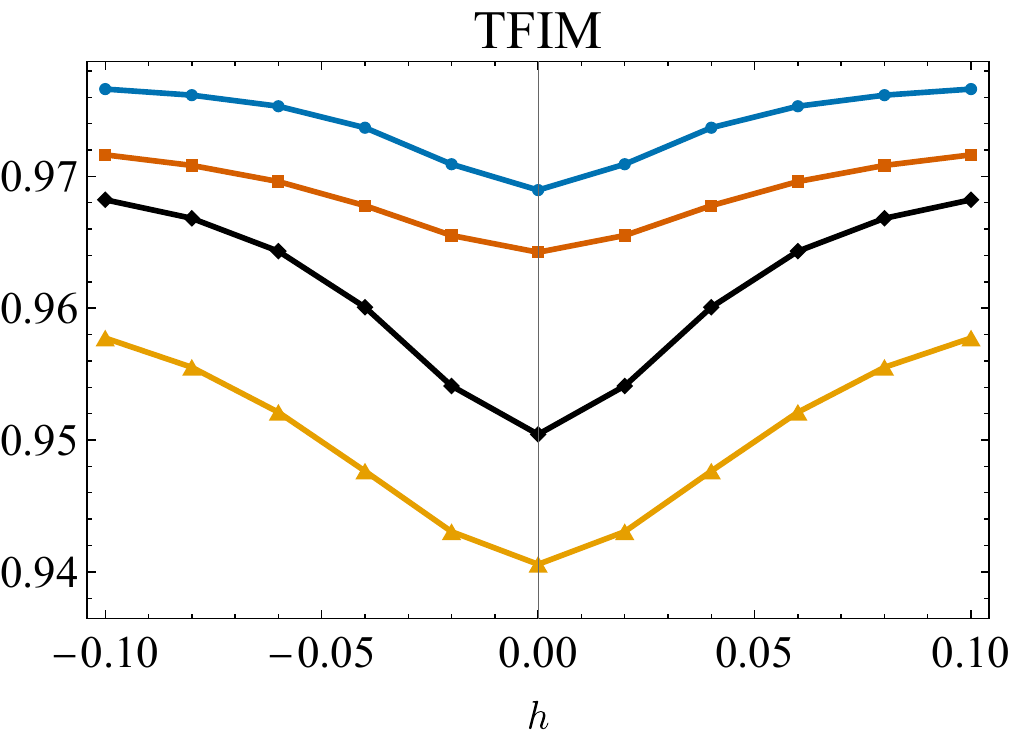}
  \caption{}
  \end{subfigure}%
  \caption{The long-time average value of $P_E(U_t)$ and $E(U_t)$ for $N=8,11$ qubits and dynamics given by (a) the Heisenberg XYZ model in an external magnetic field \cref{xyz} and (b) the transverse field Ising model \cref{tfim}. For $J_z, h \neq 0$, the models are nonintegrable and we observe that both $P_E(U_t)$ and $E(U_t)$ increase with the integrability breaking term.}
  \label{fig:models}
  \end{figure}

\nocite{readCentreRepresentationGroup1977,chengCharacterTheoryProjective2015,1980357,collinsIntegrationRespectHaar2006,robertsChaosComplexityDesign2017,goodmanSymmetryRepresentationsInvariants2009,
  christandlStructureBipartiteQuantum2006}
\end{document}


\title{Supplemental Material for \\ "An Exact Link between Nonlocal Nonstabilizerness and Operator Entanglement"}
\maketitle
\renewcommand{\theequation}{S\arabic{equation}}
\setcounter{equation}{0}
\setcounter{secnumdepth}{1}
\setcounter{lemma}{2}
\section{Proofs} \label{supplemental}
Here, we restate and provide the proofs and derivations of the results presented in the main text.
  \subsection{Theorem 1}
\begin{thm}\label{theorem}
Let $U \in \mathscr{U}_N$ and $\mathcal{H} \cong \mathcal{H}_A \otimes \mathcal{H}_B \cong (\mathbb{C}^2)^{N_A} \otimes (\mathbb{C}^2)^{\otimes N_B}$ be a bipartition of an $N$-qubit quantum system, then
\begin{equation*}
E(U^\dagger PU) = 0 \; \forall \, P \in \tilde{\mathscr{P}}_N \Leftrightarrow \exists \, V \in \mathscr{U}_{N_A}, \, W \in \mathscr{U}_{N_B}, \, C \in \mathscr{C}_N \text{ s.t. } U^\dagger=(V\otimes W) \, C
\end{equation*}
\end{thm}

  $\Leftarrow :$ Assume that $\exists\, V \in \mathscr{U}_{N_A}, \, W \in \mathscr{U}_{N_B}, \, C \in \mathscr{C}_N$ such that $U^\dagger = (V \otimes W) \, C$. Then, $\forall P \in \tilde{\mathscr{P}}_N$, $C P C^\dagger = P^\prime$, where $P^\prime \in \mathscr{P}_N$ is a Pauli string up to a global phase, hence
  \begin{equation}
  E(U^\dagger P U)= E((V \otimes W) \, C P C^\dagger (V^\dagger \otimes W^\dagger))=E((V\otimes W) P^\prime (V^\dagger \otimes W^\dagger))=E(P^\prime)=0,
  \end{equation}
where we also used that the operator entanglement is invariant under the action of local unitary channels.
\par $\Rightarrow :$ Denote as $\mathcal{L}(\mathcal{H})$ the space of linear operators acting on $\mathcal{H}$ and let $\{P_k \}_{k=1,\ldots ,d^2}$ be the set of Pauli strings on $N$ qubits. Assuming that $E(U^\dagger P_k U)=0 \; \forall \, k \in{1,\ldots, d^2}$, there are operators $X_k \in \mathcal{L}(\mathcal{H}_A) , Y_k \in \mathcal{L}(\mathcal{H}_B)$ such that
\begin{equation} \label{prod_prop}
U^\dagger P_k U = X_k \otimes Y_k \; \forall \, k.
\end{equation}
Notice that since $\{P_k\}$ spans $\mathcal{L}(\mathcal{H})$, it follows that $\{X_k\}$, $\{Y_k\}$ span $\mathcal{L}(\mathcal{H}_A)$ and $\mathcal{L}(\mathcal{H}_B)$ respectively. In addition, due to the fact that $P_k^\dagger P_k= \mathds{1}_d$ and $P_k^\dagger = P_k$, we can always choose $X_k$ and $Y_k$ to be Hermitian unitaries. To see this, first notice that
\begin{equation} \label{property_1}
(X_k^\dagger \otimes Y_k^\dagger) \, (X_k \otimes Y_k)= (U^\dagger P_k^\dagger U) \, (U^\dagger P_k U) = \mathds{1}_d \xRightarrow{\Tr_B} \Tr(Y_k^\dagger Y_k) X_k^\dagger X_k=d_B \mathds{1}_A,
\end{equation}
and
\begin{equation} \label{property_2}
{X_k}^\dagger \otimes {Y_k}^\dagger =  U^\dagger P_k^\dagger U = U^\dagger P_k U = X_k \otimes Y_k
\Rightarrow X_k^\dagger \otimes Y_k^\dagger Y_k = X_k \otimes Y_k^2 \xRightarrow{\Tr_B} \Tr(Y_k^\dagger Y_k) X_k^\dagger = \Tr(Y_k^2) X_k.
\end{equation}
Using \cref{property_2} into \cref{property_1}, we also have
\begin{equation} \label{property_3}
\Tr(Y_k^\dagger Y_k) \frac{\Tr(Y_k^2)}{\Tr(Y_k^\dagger Y_k)} X_k \frac{\Tr(Y_k^2)^*}{\Tr(Y_k^\dagger Y_k)} X_k^\dagger = d_B \mathds{1}_A \Rightarrow \left\lvert \frac{\Tr(Y_k^2)}{\Tr(Y_k^\dagger Y_k)} \right\rvert^2 = 1 \Rightarrow \frac{\Tr(Y_k^2)}{\Tr(Y_k^\dagger Y_k)} = e^{i \theta_k}.
\end{equation}
Finally, setting $\tilde{X}_k = \sqrt{\frac{\Tr(Y_k^\dagger Y_k)}{d_B}} e^{i \theta_k/2} X_k$ and $\tilde{Y}_k = \sqrt{\frac{d_B}{\Tr(Y_k^\dagger Y_k)}} e^{-i \theta_k/2} Y_k$, we have $U^\dagger P_k U = \tilde{X}_k \otimes \tilde{Y}_k$ and $\tilde{X}_k$, $\tilde{Y}_k$ are Hermitian unitaries by virtue of \cref{property_1,property_2,property_3}. Without loss of generality, we can thus rename $\tilde{X}_k \rightarrow X_k$, $\tilde{Y}_k \rightarrow Y_k$ with $X_k$, $Y_k$ Hermitian unitaries.
\par Since $P_k$ are Pauli strings, they form a projective representation of $\tilde{\mathscr{P}}_N \cong \mathbb{Z}_2^{2N} $,
\begin{equation} \label{pauli_mult}
P_{k_1} P_{k_2} = c(k_1, k_2) P_{k_1 \circ k_2}, \quad c(k_1,k_2) \in \{ \pm 1 , \pm i\}
\end{equation}
where $c(k_1,k_2): \tilde{\mathscr{P}}_N \cross \tilde{\mathscr{P}}_N \rightarrow \{\pm 1, \pm i\}$ is the 2-cocycle that keeps track of the global phases when multiplying the Pauli strings, while $\circ$ is the group operation for the Abelian group $\tilde{\mathscr{P}}_N$. Notice that the Hermiticity of the Pauli strings implies that 
\begin{equation} \label{c_prop}
c(k_2,k_1)=c(k_1, k_2)^*.
\end{equation}
Using \cref{pauli_mult}, we can then deduce multiplication rules for $X_k$ and $Y_k$:
\begin{equation} \label{calculation_1}
\begin{split}
&U^\dagger P_{k_1} U \, U^\dagger P_{k_2} U = U^\dagger (P_{k_1} P_{k_2}) U = c(k_1,k_2) U^\dagger P_{k_1 \circ k_2} U \Rightarrow
X_{k_1}\otimes Y_{k_1} \,  X_{k_2}\otimes Y_{k_2} = c(k_1,k_2) X_{k_1 \circ k_2} \otimes Y_{k_1\circ k_2} \\
&\Rightarrow X_{k_1} X_{k_2} \otimes (Y_{k_2}Y_{k_1}) (Y_{k_1}Y_{k_2})= c(k_1,k_2) X_{k_1 \circ k_2} \otimes (Y_{k_2}Y_{k_1}) Y_{k_1 \circ k_2} \\
&\xRightarrow{\Tr_B} X_{k_1}X_{k_2} = c(k_1,k_2) \frac{\Tr(Y_{k_2}Y_{k_1}Y_{k_1\circ k_2})}{d_B} X_{k_1 \circ k_2}.
\end{split}
\end{equation}
Define as $\phi(k_1,k_2) \coloneqq \frac{\Tr(Y_{k_2}Y_{k_1}Y_{k_1\circ k_2})}{d_B}$, which satisfies the property $\phi(k_2 , k_1) = \phi(k_1 ,k_2)^*$, due to the cyclicity of the trace and the Hermiticity of $Y_k$. Since the matrix product in \cref{calculation_1} is associative and $c(k_1,k_2)$ is a 2-cocycle, this means that $\phi(k_1,k_2)$ is also a 2-cocycle. In addition, since $X_{k_1}, X_{k_2}, X_{k_1 \circ k_2}$ are unitaries and $\lvert c(k_1, k_2) \rvert = 1$, we have that $\lvert \phi(k_1 ,k_2) \rvert =1$. Moreover, using $k_1 \circ k_2 = k_2 \circ k_1$, we have
\begin{equation}
  \begin{split}
&P_{k_1 \circ k_2} = P_{k_2 \circ k_1} \Rightarrow \frac{1}{c(k_1, k_2)} X_{k_1} \otimes Y_{k_1} \, X_{k_2} \otimes Y_{k_2} = \frac{1}{c(k_2 , k_1)} X_{k_2} \otimes Y_{k_2} \, X_{k_1} \otimes Y_{k_1} \\
&\Rightarrow 
X_{k_1} X_{k_2} \otimes (Y_{k_2} Y_{k_1}) Y_{k_1} Y_{k_2} = \frac{c(k_1,k_2)}{c(k_2 , k_1)} X_{k_2} X_{k_1} \otimes (Y_{k_2} Y_{k_1}) Y_{k_2} Y_{k_1}\\
&\xRightarrow{\Tr_B} 
X_{k_1} X_{k_2} = \frac{c(k_1,k_2)}{c(k_2 , k_1)} \frac{\Tr(Y_{k_2} Y_{k_1} Y_{k_2} Y_{k_1})}{d_B} X_{k_2} X_{k_1}\\
&\xRightarrow{\lambda (k_1 , k_2) \coloneqq \frac{\Tr(Y_{k_2} Y_{k_1} Y_{k_2} Y_{k_1})}{d_B} \in \mathbb{R}}
\phi(k_1,k_2) c(k_1,k_2) X_{k_1 \circ k_2} = \frac{c(k_1,k_2)}{c(k_2 , k_1)} \lambda (k_1 , k_2) \phi(k_2,k_1) c(k_2 ,k_1) X_{k_2 \circ k_1} \\
&\Rightarrow \phi(k_1 ,k_2) = \lambda(k_1 ,k_2) \phi(k_2 ,k_1) = \lambda(k_1 ,k_2) \phi(k_1 ,k_2)^*.
\end{split}
\end{equation}
Since $\lvert \phi(k_1 , k_2) \rvert =1 \Rightarrow \lvert \lambda(k_1 , k_2) \rvert =1 \xRightarrow{\lambda(k_1 , k_2) \in \mathbb{R}} \lambda(k_1 ,k_2) = \pm 1$, which means that $\phi(k_1 ,k_2) = \pm \phi(k_1 ,k_2)^*$, so that in the end $\phi(k_1 ,k_2) \in \{\pm 1 , \pm i\}$. All in all, we have
\begin{align}\label{X_mult}
&X_{k_1} X_{k_2} = c(k_1,k_2) \phi(k_1,k_2) X_{k_1 \circ k_2}, \\
& \phi(k_1,k_2) \in \{ \pm 1 , \pm i\} \text{ and } \phi(k_2,k_1) = \phi(k_1,k_2)^* \label{phi_prop}
\end{align}
where $c(k_1 ,k_2) \in \{\pm 1 , \pm i \}$ is defined by \cref{pauli_mult}. The corresponding relation for $Y_k$, compatible with \cref{prod_prop,pauli_mult,X_mult}, is
\begin{equation} \label{Y_mult}
Y_{k_1} Y_{k_2} = \frac{1}{\phi(k_1,k_2)} Y_{k_1 \circ k_2}.
\end{equation}
Notice that from \cref{X_mult,Y_mult}, it follows that both $\{X_k\}$ and $\{Y_k\}$ form projective representations of $\tilde{\mathscr{P}}_N$ on $\mathcal{H}_A$ (with dimension $2^{N_A}$ and 2-cocycle $c\cdot \phi$) and $\mathcal{H}_B$ (with dimension $2^{N_B}$ and 2-cocycle $1/\phi$), respectively. Since $\{X_k\}$ and $\{Y_k\}$ span $\mathcal{L}(\mathcal{H}_A)$ and $\mathcal{L}(\mathcal{H}_B)$, respectively, these projective representations are irreducible. Before proceeding, we recall some known results from the projective representation theory of finite Abelian groups.
\begin{definition} \label{regularity}
  Let $G$ be a finite Abelian group, and $c: G \cross G \rightarrow \mathbb{C}$ be a 2-cocycle of $G$. We say that an element $g\in G$ is $c-$regular if $c(g,x)=c(x,g) \; \forall \, x \in G$. 
\end{definition}
\begin{lemma} \label{coboundary}
  Let $G$ be a finite Abelian group, and $c$ be a 2-cocycle of $G$. The $c$-regular elements of $G$ form a subgroup $G_0$ and the restriction of $c$ on $G_0$ is a coboundary.
\end{lemma}
  \begin{proof}
    For an Abelian group the set of $c$-regular elements is $\{g_0 \in G \vert c(g_0,x)=c(x,g_0) \; \forall \, x \in G\}$. Then, by Lemma 3.3 of Ref. \cite{readCentreRepresentationGroup1977} this is a subgroup of $G$. In addition, for any two elements $g_0,g_0^\prime \in G_0$, $c(g_0,g_0^\prime)=c(g_0^\prime , g_0)$, hence $G_0$ is abelian and $c$-symmetric. Then, by Lemma 2.12 of Ref. \cite{chengCharacterTheoryProjective2015}, $c$ is a coboundary on $G_0$.
    \end{proof}
    \begin{lemma}[page 380, Theorem 2.21 of Ref. \cite{1980357}] \label{degree}
      Let $G$ be a finite Abelian group, $c$ be a 2-cocycle of $G$ and $G_0$ be the subgroup of $c$-regular elements of $G$. Then, all irreducible projective $c$-representations of $G$ have degree equal to $\sqrt{\lvert G \rvert / \lvert G_0 \rvert}$ and are projectively equivalent.
    \end{lemma}
    In our case, we have a $2^{N_A}$-dimensional projective $c\cdot \phi$-irrep and a $2^{N_B}$-dimensional projective $1/\phi$-irrep of $\tilde{\mathscr{P}}_N$, so we set $G=\tilde{\mathscr{P}}_N \cong \mathbb{Z}_2^{2N}$. Then, \cref{degree} implies that the subgroups $H$, $\tilde{H}$ of $c\cdot \phi$-regular and $1/\phi$-regular elements, respectively, in $\tilde{\mathscr{P}}_N$ have orders $\lvert H \rvert = 2^{2N_B}$ and $\lvert \tilde{H} \rvert = 2^{2N_A}$, respectively. Notice that due to \cref{c_prop,phi_prop} and the fact that $c,\phi \in \{ \pm 1 , \pm i \}$, the $c \cdot \phi$- and $1/\phi-$regularity imply that
    \begin{align}
    &c(h,g) \phi(h,g) = \pm 1 \; \forall \, h \in H, \, g \in G \label{regularity_1}\\
    & \phi(\tilde{h},g)= \pm 1 \; \forall \, \tilde{h} \in \tilde{H}, \, g \in G. \label{regularity_2}
    \end{align}
    Combining \cref{regularity_1,regularity_2} with \cref{X_mult,Y_mult} and using the fact that $X_k,Y_k$ are Hermitian, we then have
    \begin{align}
      &X_h X_g = \pm X_{h \circ g} \Rightarrow [X_h , X_g] =0 \; \forall \, h \in H, \, g \in G \Rightarrow X_h=\tau_1(h) \mathds{1}_A, \; \tau_1(h)= \pm 1 \; \forall \, h \in H, \label{identity_x}\\
      &Y_{\tilde{h}} Y_g = \pm Y_{\tilde{h} \circ g} \Rightarrow [Y_{\tilde{h}} , Y_g] =0 \; \forall \, \tilde{h} \in \tilde{H}, \, g \in G \Rightarrow Y_{\tilde{h}}=\lambda_1(\tilde{h}) \mathds{1}_B, \; \lambda_1(\tilde{h})= \pm 1\; \forall \, \tilde{h} \in \tilde{H}, \label{identity_y}
      \end{align}
      where we used Schur's lemma and the fact that $X_h$, $Y_{\tilde{h}}$ are Hermitian unitaries.
      Also, using \cref{coboundary}, we have that:
      \begin{equation}
      \text{The restrictions of } c\cdot \phi \text{ on } H \text{ and } \frac{1}{\phi} \text{ on } \tilde{H} \text{ are coboundaries}.
      \end{equation}
      Since $1/\phi$ is a coboundary on $\tilde{H}$, it follows that $\phi$ is also a coboundary on $\tilde{H}$. Then, from \cref{X_mult}, we find that $\{X_{\tilde{h}} \, \vert \, \tilde{h} \in \tilde{H}\}$ forms a $2^{N_A}$-dimensional projective $c\cdot \phi-$representation $\pi$ of $\tilde{H}$. Since $\phi$ is a coboundary on $\tilde{H}$ this representation is projectively equivalent to a $2^{N_A}$-dimensional $c$-representation $\pi^\prime$ of $\tilde{H}$.
      However, we know that the Pauli strings $\{P_\alpha \, \vert \, \alpha=1,\ldots,2^{2N_A} \}$ on $N_A$ qubits form a $2^{N_A}$-dimensional projective $c$-irrep of $\mathbb{Z}_2^{2N_A}$ and by \cref{degree} all $c$-irreps have the same dimension. So, $\pi^\prime$ has to be irreducible as well, since it cannot be broken down into smaller irreps, and by \cref{degree} is projectively equivalent to the Pauli strings projective $c$-irrep. All in all, we have that $\{X_{\tilde{h}} \, \vert \, \tilde{h} \in \tilde{H}\}$ is projectively equivalent to Pauli strings on $N_A$ qubits, namely
      \begin{equation} \label{partial_x}
      X_{\tilde{h}} = \lambda_2(\tilde{h}) V P_{\alpha(\tilde{h})} V^{-1},
      \end{equation}
      where $\lambda_2 \in \mathbb{C}$ and $V$ is an invertible operator. Since, $X_{\tilde{h}}$ and $P_\alpha$ are unitaries, $V$ is also a unitary operator and $\lvert \lambda_2 \rvert=1$. In addition, since $X_{\tilde{h}}$ and $P_\alpha$ are Hermitian, $\lambda_2 \in \mathbb{R}$, so that in the end $\lambda_2 = \pm 1$.
      \par Similarly, since $c\cdot \phi$ is a coboundary on $H$, it follows that $\frac{1}{\phi}$ is cohomologous to $c$, namely it is equal to $c$ up to a coboundary. Then, using \cref{Y_mult} and similar arguments as before, we deduce that $\{Y_h \, \vert \, h \in H\}$ forms a projective $1/\phi$-irrep of $H$ which is projectively equivalent to Pauli strings on $N_B$ qubits, such that
      \begin{equation} \label{partial_y}
      Y_h = \tau_2(h) W P_{\beta(h)} W^{-1},
      \end{equation}
      where $W$ is a unitary operator and $\tau_2(h) = \pm 1$. Defining $\lambda(\tilde{h})=\lambda_1(\tilde{h}) \lambda_2(\tilde{h}) = \pm 1$ and $\tau(h)=\tau_1(h) \tau_2(h) = \pm 1$ and combining \cref{prod_prop,identity_x,identity_y,partial_x,partial_y}, we find that
      \begin{equation}
        \begin{split}
        &U^\dagger P_{\tilde{h}} U = \lambda(\tilde{h}) V P_{\alpha(\tilde{h})} V^\dagger \otimes \mathds{1}_B \; \forall \tilde{h} \in \tilde{H} \\
        &U^\dagger P_h U = \tau(h) \mathds{1}_A \otimes W P_{\beta(h)} W^\dagger \; \forall h \in H.
        \end{split}
        \end{equation}

        Notice that if $q \in H \cap \tilde{H}$, \cref{regularity_1,regularity_2} imply that $c(q,g)=\pm 1 \; \forall \, g\in G$, which means that $P_q=\mathds{1}$, i.e., $q=e$ is the identity group element. In addition, $\lvert G \rvert = \lvert \tilde{H} \rvert \, \lvert H \rvert$, so we can write any element of $G$ uniquely as $g = \tilde{h} \circ h$ with $\tilde{h} \in \tilde{H}$ and $h \in H$, namely $G= \tilde{H} \cross H$. Using \cref{prod_prop,partial_x,partial_y,identity_x,identity_y}, we then have
        \begin{equation} \label{mid_step}
        U^\dagger P_{\tilde{h}} P_h U = \lambda_{\tilde{h}} \tau_h VP_{\alpha(\tilde{h})} V^\dagger \otimes W P_{\beta(h)} W^\dagger \; \forall \tilde{h} \in \tilde{H}, \, h \in H.
        \end{equation}
        Using \cref{regularity_1,regularity_2} once more we find that for any $h \in H$ and $\tilde{h} \in \tilde{H}$, $c(h,\tilde{h}) = \pm 1$, i.e. $[P_{\tilde{h}}, P_h]=0$. So, we can always find a Clifford unitary $C$ such that
        \begin{equation} \label{clifford_map}
        \begin{split}
        &C P_{\tilde{h}} C^\dagger= \lambda_{\tilde{h}} P_{\alpha(\tilde{h})} \otimes \mathds{1}_B\; \forall \tilde{h} \in \tilde{H}, \\
        & C P_h C^\dagger = \tau_h \mathds{1}_A \otimes P_{\beta(h)} \; \forall h \in H,
        \end{split}
        \end{equation}
        since this amounts to, respectively, mapping the Pauli strings $P_{\tilde{h}}$ and $P_h$ to Pauli strings on $N_A$ and $N_B$ distinct qubits, and choosing the corresponding signs $\pm 1$. So, combining \cref{mid_step,clifford_map}, we get 
        \begin{equation}
        U^\dagger P_{\tilde{h}} P_h U = (V\otimes W) \, C \, (P_{\tilde{h}} P_h) \, C^\dagger \, (V^\dagger \otimes W^\dagger) \; \forall \tilde{h} \in \tilde{H}, \, h \in H.
        \end{equation}
        Since $P_{\tilde{h}} P_h$ above run through all the Pauli strings in $\tilde{\mathscr{P}}_N$, which form an operator basis, this means that
        \begin{equation}
        U^\dagger = (V \otimes W) \, C.
        \end{equation}
        \qed

    \subsection{Lemma 1}

\setcounter{lemma}{0}
\begin{lemma}\label{lemma_no_NL}
For all unitaries $U \in \mathscr{U}_{N}$ s.t. $U= (V \otimes W) \, C$, where $V \in \mathscr{U}_{N_A}, \, W \in \mathscr{U}_{N_B}, \, C \in \mathscr{C}_N$,
\begin{equation} \label{no_NL}
m^{\text{NL}}(U \ket{\psi}) \leq m(\ket{\psi}).
\end{equation}
In particular, if $\ket{s}$ is a stabilizer state, then $m^{\text{NL}}(U \ket{s})=0$.
\end{lemma}
\setcounter{lemma}{4}
    For $U=(V \otimes W) C$, we have
    \begin{equation}
    m^{\text{NL}}(U\ket{\psi})= \min_{\substack{U_A \in \mathscr{U}_{N_A} \\ U_B \in \mathscr{U}_{N_B}}} m((U_A \otimes U_B) (V \otimes W) C \ket{\psi}) \leq m(C \ket{\psi}) = m(\psi),
    \end{equation}
    where in the last step we used that any measure of nonstabilizerness is invariant under Clifford unitaries.

    \subsection{Proposition 1}
    
 \begin{proposition}\label{prop_paul_ent}
Let $U \in \mathscr{U}_N$ and $\mathcal{H} \cong \mathcal{H}_A \otimes \mathcal{H}_B \cong (\mathbb{C}^2)^{N_A} \otimes (\mathbb{C}^2)^{\otimes N_B}$, then the average Pauli-entangling power is given as
\begin{equation} \label{paul_ent}
P_E(U) \coloneqq \mathlarger{\mathbb{E}}_{P \in \tilde{\mathscr{P}}_N} E_{\text{lin}}(U^\dagger P U) = 1- \frac{1}{d^2} \Tr( T^A_{(12)(34)} \, {U^\dagger}^{\otimes 4} Q U^{\otimes 4}),
\end{equation}
where $Q \coloneqq d^{-2} \sum_{P \in \tilde{\mathscr{P}_N}} P^{\otimes 4}$.
\end{proposition}
    Using Eq.~(1) of the main text, we have
    \begin{equation} \label{starting_step}
    P_E(U)=1-\frac{1}{d^4} \sum_{P\in \tilde{\mathscr{P}}_N}\Tr(T_{12}^A \, {U^\dagger}^{\otimes 2} \, P^{\otimes 2} \, U^{\otimes 2} \, T_{12}^A \, {U^\dagger}^{\otimes 2} \, P^{\otimes 2} \, U^{\otimes 2}).
    \end{equation}
Notice that the algebra $\mathcal{K}=\Span\{P^{\otimes 2} \, \vert \, P \in \tilde{\mathscr{P}}_N \}$ is a $d^2$-dimensional Abelian algebra acting on the $d^2$-dimensional Hilbert space $\mathcal{H}^{\otimes 2}$. This means that $\mathcal{K}$ is maximally Abelian, i.e., $\mathcal{K}^\prime = \mathcal{K}$, where $\mathcal{K}^\prime$ is the commutant algebra of $\mathcal{K}$. Then,
\begin{equation} \label{pauli_average}
\frac{1}{d^2} \sum_{P \in \tilde{\mathscr{P}}_N} P^{\otimes 2} (\bullet ) P^{\otimes 2} = \mathbb{P}_{\mathcal{K}^\prime} = \mathbb{P}_{\mathcal{K}} = \sum_{P \in \tilde{\mathscr{P}}_N} \Tr(\frac{P^{\otimes 2}}{d} (\bullet )) \frac{P^{\otimes 2}}{d},
\end{equation}
where $\mathbb{P}_{\mathcal{K}}$ is the projector onto $\mathcal{K}$ and we used that $\{P/d \, \vert \, P \in \tilde{\mathscr{P}}_N \}$ is an orthonormal basis of $\mathcal{K}$.
Using \cref{pauli_average} in \cref{starting_step}, we get
\begin{equation} \label{step_1}
  \begin{split}
P_E(U)&=1- \frac{1}{d^4} \sum_{P \in \tilde{\mathscr{P}}_N} \Tr( T_{12}^A \, {U^\dagger}^{\otimes 2} \, \Tr(P^{\otimes 2} \, U^{\otimes 2} \, T_{12}^A \, {U^\dagger}^{\otimes 2}) \, P^{\otimes 2} \, U^{\otimes 2}) \\
&= 1 - \frac{1}{d^4} \sum_{P \in \tilde{\mathscr{P}}_N} \left( \Tr(P^{\otimes 2} \, U^{\otimes 2} \, T_{12}^A \, {U^\dagger}^{\otimes 2}) \right)^2 = 1- \frac{1}{d^4} \sum_{P \in \tilde{\mathscr{P}}_N} \Tr(T_{(12)(34)}^A \, {U^\dagger}^{\otimes 4} \, P^{\otimes 4} \, U^{\otimes 4})\\
&= 1 - \frac{1}{d^2} \Tr(T_{(12)(34)}^A \, {U^\dagger}^{\otimes 4} \, Q \, U^{\otimes 4}),
  \end{split}
\end{equation}
  where $Q \coloneqq d^{-2} \sum_{P \in \tilde{\mathscr{P}}_N} P^{\otimes 4}$. It is easy to see that $Q$ is a $d^2$-dimensional projector. In fact, we can write explicitly an orthonormal basis of $Q$ as $\{\ket{\psi_{\vec{x}\vec{z}}}^{\otimes 2} \, \vert \, \vec{x},\vec{z} \in \{0,1\}^N \}$, where $\ket{\psi_{\vec{x}\vec{z}}} = 2^{-N/2} \sum_{\vec{y} \in \{0,1\}^N} (-1)^{\vec{y}\cdot \vec{z}} \ket{\vec{y}} \otimes \ket{\vec{y} \oplus \vec{x}} = Z^{\vec{z}} \otimes X^{\vec{x}} \ket{\phi^+}$. Here $Z^{\vec{z}}$ and $X^{\vec{x}}$ describe $Z$ and $X$ Pauli strings based on the value of the bitstrings $\vec{z}$, $\vec{x}$ and $\ket{\phi^+} = 2^{-N/2} \sum_{\vec{y} \in \{0,1\}^N} \ket{\vec{y}} \otimes \ket{\vec{y}}$ is a maximally entangled state.

  \subsection{Lemma 2}
\setcounter{lemma}{1}
  \begin{lemma}\label{lemma_opcoh}
For any unitary operator $U \in \mathscr{U}_N$,
\begin{equation} \label{opcoh}
E_{\text{lin}}(U)= \min_{{\substack{V_A, W_A \in \mathscr{U}_{N_A} \\ V_B, W_B \in \mathscr{U}_{N_B}}}} M_{\text{lin}}(V_A \otimes V_B \, U \, W_A \otimes W_B).
\end{equation}
\end{lemma}
\setcounter{lemma}{4}
  This follows directly from combining Observation 1 of the main text with Theorem 2 of Ref. \cite{anandQuantumCoherenceSignature2021}, which we state below as a Lemma.
  \begin{lemma}[Theorem 2 of Ref. \cite{anandQuantumCoherenceSignature2021}] \label{lemma_coh}
    Let $\ket{\psi}$ be a quantum state in $\mathcal{H} \cong \mathcal{H}_A \otimes \mathcal{H}_B$ and $c^{(2)}_{\mathbb{B}}(\ketbra{\psi}{\psi}) = 1-\sum_j \lvert \braket{j}{\psi} \rvert^4$ be the 2-coherence with respect to the basis $\mathbb{B}=\{ \ket{j} \}_{j=1}^d$. Then,
    \begin{equation}
     \min_{\mathbb{B}_a, \mathbb{B}_b} c^{(2)}_{\mathbb{B}_a \otimes \mathbb{B}_b} (\ketbra{\psi}{\psi}) = S_{\text{lin}}(\Tr_B(\ketbra{\psi}{\psi})),
    \end{equation}
    where $S_{\text{lin}}(\rho)= 1 - \Tr(\rho^2)$ is the linear entropy of the reduced state $\rho$.
  \end{lemma}
  Notice that the minimization over the local bases $\mathbb{B}_a$ and $\mathbb{B}_b$ is equivalent to the minimization over local unitaries $U_A$ and $U_B$ and some fixed tensor product basis $\mathbb{B}_1 \otimes \mathbb{B}_2$. Using Observation 1, we can recast \cref{lemma_coh} in operator space yielding \cref{opcoh}, where now the role of the state is played by a unitary operator $U$ and the local unitaries are now local unitary superoperators $V_A (\bullet) W_A$ and $V_B (\bullet) W_B$.

  \subsection{Corollary 1}
  \begin{cor} \label{cor_no_op_NL}
For any $U \in \mathscr{U}_N$,
\begin{equation} \label{no_op_NL}
P_E(U) = 0 \Leftrightarrow \, \exists \, V \in \mathscr{U}_{N_A}, \, W \in \mathscr{U}_{N_B} \text{ s.t. } \forall P \in \tilde{\mathscr{P}}_N: M_{\text{lin}}(V^\dagger \otimes W^\dagger \, U^\dagger PU \, V \otimes W) =0.
\end{equation}
\end{cor}
  We have that $P_E(U) =0 \Leftrightarrow E_{\text{lin}}(U^\dagger P U) =0 \; \forall \, P \in \tilde{\mathscr{P}}_N$. From Theorem 1, this happens if and only if $U^\dagger = (V \otimes W) \, C$ for some $V \in \mathscr{U}_A, W\in \mathscr{U}_B, C \in \mathscr{C}_N$. Assuming this is true,
  \begin{equation}
  M_{\text{lin}}((V^\dagger \otimes W^\dagger) \, U^\dagger P U \, (V \otimes W)) = M_{\text{lin}}(C P C^\dagger)= M_{\text{lin}}(P) =0 \; \forall \, P \in \tilde{\mathscr{P}}_N.
  \end{equation}
  For the converse, notice that if  $M_{\text{lin}}((V^\dagger \otimes W^\dagger) \, U^\dagger P U \, (V \otimes W)) = 0 \; \forall \, P \in \tilde{\mathscr{P}}_N$ for some $V \in \mathscr{U}_A, W\in \mathscr{U}_B$, then Lemma 2 immediately implies that $E_{\text{lin}}(U^\dagger P U) =0 \; \forall \, P \in \tilde{\mathscr{P}}_N$.

  \subsection{Proposition 2}
  \begin{proposition} \label{prop_upbound}
For any $U \in \mathscr{U}_N$,
\begin{equation} \label{upbound}
P_E(U) \leq \min\left\{ {\mathlarger{\mathbb{E}}}_{P_A \in \tilde{\mathscr{P}}_{N_A}} M_{\text{lin}}(U \, P_A \otimes \mathds{1}_B \, U^\dagger),\, {\mathlarger{\mathbb{E}}}_{P_B \in \tilde{\mathscr{P}}_{N_B}} M_{\text{lin}}(U \, \mathds{1}_A \otimes P_B \, U^\dagger) \right\},
\end{equation}
\end{proposition}
  Using the identity $T_{12}^A=d_A^{-1} \sum_{P_A \in \tilde{\mathscr{P}}_{N_A}} P_A \otimes \mathds{1}_{B_1} \otimes P_A \otimes \mathds{1}_{B_2}$, we can rewrite \cref{step_1} as
  \begin{equation} \label{step_2}
    \begin{split}
    P_E(U) &= 1 - \frac{1}{d^4} \sum_{P \in \tilde{\mathscr{P}}_{N}} \left( \Tr( \frac{1}{d_A} \sum_{P_A \in \tilde{\mathscr{P}}_{N_A}} (P_A \otimes \mathds{1}_B)^{\otimes 2} \, {U^\dagger}^{\otimes 2} \, P^{\otimes 2} \, U^{\otimes 2}) \right)^2 \\
    &= 1-\frac{1}{d^4 d_A^2} \sum_{P \in \tilde{\mathscr{P}}_{N}} \left(\sum_{P_A \in \tilde{\mathscr{P}}_{N_A}} \left( \Tr( P_A \otimes \mathds{1}_B U^\dagger P U) \right)^2 \right)^2 \\
    &= 1 - \frac{1}{d_A^2} \sum_{\substack{ P \in \tilde{\mathscr{P}}_{N} \\ P_A,P_A^\prime \in \tilde{\mathscr{P}}_{N_A}}} \Xi_P^{P_A}(U) \, \Xi_P^{P_A^\prime}(U),
    \end{split}
    \end{equation}
    where we defined $\Xi_P^{P_A}(U) \coloneqq \frac{1}{d^2} \left(\Tr( (P_A \otimes \mathds{1}_B) \, {U^\dagger} \, P \, U)\right)^2$, which is a probability distribution over $\tilde{\mathscr{P}}_N$ for any $U \in \mathscr{U}_N$ and $P_A \in \tilde{\mathscr{P}}_{N_A}$. Then, we can bound \cref{step_2} as
    \begin{equation}
      \begin{split}
    P_E(U) & \leq 1 - \frac{1}{d_A^2} \sum_{\substack{ P \in \tilde{\mathscr{P}}_{N} \\ P_A \in \tilde{\mathscr{P}}_{N_A}}} \left(\Xi_P^{P_A}\right)^2 = \mathlarger{\mathbb{E}}_{P_A \in \tilde{\mathscr{P}}_{N_A}} \left(1- \frac{1}{d^4} \sum_{P \in \tilde{\mathscr{P}}_{N}} \left( \Tr( (P_A \otimes \mathds{1}_B) \,  U^\dagger P U) \right)^4 \right) \\
    &= \mathlarger{\mathbb{E}}_{P_A \in \tilde{\mathscr{P}}_{N_A}} M_{\text{lin}} \left(U \, (P_A \otimes \mathds{1}_B) \, U^\dagger \right).
      \end{split}
    \end{equation}
    Similarly, we can obtain an upper bound by exchanging $A$ with $B$, which proves \cref{upbound}.

    \subsection{Proposition 3}
 \begin{proposition} \label{prop_typical}
  \begin{equation} \label{typical}
  {\mathlarger{\mathbb{E}}}_{U \in \mathscr{U}_N} P_E(U)=\frac{(d^2-d_A^2)(d^2-10)(d_A^2-1)}{d^2 d_A^2(d^2-9)}= 1-\left(1-\frac{1}{d^2}\right)\frac{1}{d_A^2}-\left(1-\frac{1}{d^2}\right)\frac{1}{d_B^2} + O\left(\frac{1}{d^4}\right).
  \end{equation}
  \end{proposition}
    The calculation of the Haar average in \cref{typical} follows from the application of the Weingarten calculus \cite{collinsIntegrationRespectHaar2006,robertsChaosComplexityDesign2017}. The general statement for a Haar average over the unitary group is
    \begin{equation}
        {\mathlarger{\mathbb{E}}}_{U \in \mathscr{U}_N} {U^\dagger}^{\otimes q} (\bullet ) U^{\otimes q} = \sum_{\pi, \sigma \in \mathcal{S}_q} W_{\pi , \sigma} \Tr( T_\sigma (\bullet )) T_\pi,
    \end{equation}
    where $\mathcal{S}_q$ is the symmetric group over $q$ elements and $T$ is the representation map from $\mathcal{S}_q$ to $\mathcal{L}(\mathcal{H}^{\otimes q})$. The Weingarten function $W_{\pi , \sigma}$ is given as
    \begin{equation} \label{weingarten}
    W_{\pi , \sigma} = \frac{1}{(q!)^2}\sum_{\lambda \vdash q} \frac{\chi_\lambda^2(e) \, \chi_\lambda(\pi \sigma)}{d_\lambda},
    \end{equation}
    where $\lambda \vdash q$ denotes that $\lambda$ is an integer partition of $q$, which characterizes the irreducible representations of $\mathcal{S}_q$ on $\mathcal{H}^{\otimes q}$, $\chi_\lambda$ is the character of the irreducible representation of $\mathcal{S}_q$ corresponding to $\lambda$, $e$ is the identity element of $\mathcal{S}_q$ and $d_\lambda$ is the dimension of the irreducible representation of $\mathscr{U}_N$ corresponding to $\lambda$ in the Schur-Weyl decomposition \cite{goodmanSymmetryRepresentationsInvariants2009,christandlStructureBipartiteQuantum2006}, $\mathcal{H}^{\otimes q} = \bigoplus_{\lambda \vdash q} \mathbb{C}_{\chi_\lambda(e)} \otimes \mathbb{C}_{d_\lambda}$. While the above statements are fairly standard in the literature, we believe it is useful to provide an alternate derivation from the original one of Ref. \cite{collinsIntegrationRespectHaar2006}, based on the arguments of Ref. \cite{robertsChaosComplexityDesign2017} and basic facts from representation theory. The starting point is Schur-Weyl duality which asserts that the commutant of the algebra $\mathcal{A} = \Span\{U^{\otimes q} \, \vert \, U \in \mathscr{U}_N \}$ is $\mathcal{A}^\prime = \Span\{T_\pi \, \vert \, \pi \in \mathcal{S}_q \}$. Then,
    \begin{equation}
      {\mathlarger{\mathbb{E}}}_{U \in \mathscr{U}_N} {U^\dagger}^{\otimes q} (\bullet ) U^{\otimes q} = \mathbb{P}_{\mathcal{A}^\prime} (\bullet ) = \sum_{\pi \in \mathcal{S}_q} b_\pi( \bullet ) T_\pi,
    \end{equation}
    where $b_\pi$ is a linear functional $\mathcal{L}(\mathcal{H}^{\otimes q}) \rightarrow \mathbb{C}$, so $\exists \, R_\pi \in \mathcal{L}(\mathcal{H}^{\otimes q})$ such that $b_\pi(\bullet) = \Tr(R_\pi \bullet)$. Notice though that using the right invariance of the Haar measure, for any $V \in \mathscr{U}_N$, we have
    \begin{equation}
      \begin{split}
      &{\mathlarger{\mathbb{E}}}_{U \in \mathscr{U}_N} {U^\dagger}^{\otimes q} {V^\dagger}^{\otimes q} (\bullet ) V^{\otimes q} U^{\otimes q} = {\mathlarger{\mathbb{E}}}_{U \in \mathscr{U}_N} {U^\dagger}^{\otimes q} (\bullet ) U^{\otimes q} \Rightarrow \sum_{\pi \in \mathcal{S}_q} \Tr(R_\pi {V^\dagger}^{\otimes q} (\bullet ) V^{\otimes q}) T_\pi = \sum_{\pi \in \mathcal{S}_q} \Tr(R_\pi (\bullet ) ) T_\pi \\
      &\xRightarrow{\text{lin. independence}} \Tr(V^{\otimes q} R_\pi {V^\dagger}^{\otimes q} (\bullet ))= \Tr(R_\pi (\bullet )) \; \forall \, \pi \in \mathcal{S}_q \Rightarrow V^{\otimes q} R_\pi {V^\dagger}^{\otimes q} = R_\pi \; \forall \, \pi \in \mathcal{S}_q.
      \end{split}
      \end{equation}

    This means that $R_\pi \in \mathcal{A}^\prime$, i.e., $R_\pi = \sum_{\sigma \in \mathcal{S}_q} W_{\pi , \sigma} T_\sigma$. In addition, for any $\pi^\prime \in \mathcal{S}_q$,
    \begin{equation} \label{step_return}
      \begin{split}
    &{\mathlarger{\mathbb{E}}}_{U \in \mathscr{U}_N} {U^\dagger}^{\otimes q} T_{\pi^\prime} U^{\otimes q} = T_{\pi^\prime} \Rightarrow \sum_{\pi, \sigma \in \mathcal{S}_q} W_{\pi , \sigma} \Tr(T_{\pi^\prime} T_\sigma) T_\pi = T_{\pi^\prime} \\
    &\xRightarrow{\text{lin. independence}} \sum_{\sigma \in \mathcal{S}_q} W_{\pi , \sigma} \, \chi(\pi^\prime \sigma) = \delta_{\pi , \pi^\prime},
      \end{split}
    \end{equation}
    where $\Tr(T_{\pi^\prime} T_\sigma)=\chi(\pi^\prime \sigma) = \sum_{\lambda \vdash q} d_\lambda \chi_\lambda(\pi^\prime \sigma)$. So, to find $W_{\pi , \sigma}$ we need to invert $F_{\sigma , \pi} \coloneqq \chi(\pi \sigma)$. Since $F$ is a function of the conjugacy class $[\pi \sigma]$, it follows that $W_{\pi , \sigma}$ is also a function of the conjugacy class $[\pi \sigma]$, so
    \begin{equation} \label{weingarten_step}
    W_{\pi , \sigma} = \sum_{\lambda^\prime \vdash q} c_{\lambda^\prime} \chi_{\lambda^\prime}(\pi \sigma).
    \end{equation}
    Returning to \cref{step_return}, we have 
    \begin{equation} \label{step_return_2}
    \sum_{\sigma \in \mathcal{S}_q} \sum_{\lambda,\lambda^\prime \vdash q} c_{\lambda^\prime} \, d_\lambda \, \chi_{\lambda^\prime}(\pi \sigma) \, \chi_\lambda(\pi^\prime \sigma) = \delta_{\pi, \pi^\prime}.
    \end{equation}
    Notice that if $\rho_\lambda$ is the irreducible representation of $\mathcal{S}_q$ corresponding to $\lambda$, we have that 
    \begin{equation}
      \begin{split}
    \sum_{\sigma \in \mathcal{S}_q} \chi_\lambda(\pi \sigma) \, \chi_{\lambda^\prime}(\pi^\prime \sigma) &= \sum_{\sigma \in \mathcal{S}_q} \sum_{i,j=1}^{\chi_\lambda(e)} \sum_{k,l=1}^{\chi_{\lambda^\prime}(e)} [\rho_\lambda(\pi)]_{ij} [\rho_\lambda(\sigma)]_{ji}  \, [\rho_{\lambda^\prime}(\pi^\prime)]_{kl} [\rho_{\lambda^\prime}(\sigma)]_{lk} \\
    &= \sum_{i,j=1}^{\chi_\lambda(e)} \sum_{k,l=1}^{\chi_{\lambda^\prime}(e)} [\rho_\lambda(\pi)]_{ij} [\rho_{\lambda^\prime}(\pi^\prime)]_{kl} \delta_{\lambda , \lambda^\prime} \, \delta_{j,l} \delta_{i,k} \frac{q!}{\chi_\lambda(e)} = \frac{q!}{\chi_\lambda(e)} \sum_{i,j=1}^{\chi_\lambda(e)} [\rho_\lambda(\pi)]_{ij} [\rho_{\lambda}(\pi^\prime)]_{ij} \\
    &= \frac{q!}{\chi_\lambda(e)} \chi_\lambda(\pi {\pi^\prime}^{-1})
      \end{split}
    \end{equation}
    where we used Schur's orthogonality theorem for the matrix elements of the irreducible representation $\rho_\lambda$ of $\mathcal{S}_q$ and that $[\rho_\lambda(\pi^\prime)]_{ij}= [\rho_\lambda({\pi^\prime}^{-1})]_{ji}$ . So, \cref{step_return_2} becomes
    \begin{equation}
      \begin{split}
     &q! \, \sum_{\lambda \vdash q} c_{\lambda} \, d_\lambda \frac{\chi_\lambda(\pi {\pi^\prime}^{-1})}{\chi_\lambda(e)}= \delta_{\pi, \pi^\prime} = \frac{1}{q!} \sum_{\lambda \vdash q} \chi_\lambda(e) \chi_\lambda(\pi {\pi^\prime}^{-1}) \\
     &\xRightarrow{\text{lin. independence}} q! \, \frac{ c_{\lambda} \, d_\lambda}{\chi_\lambda(e)}= \frac{1}{q!} \chi_\lambda(e)  \Rightarrow c_{\lambda} = \frac{\chi_\lambda(e)^2}{(q!)^2 d_\lambda},
      \end{split}
    \end{equation}
    which together with \cref{weingarten_step} proves \cref{weingarten}.
    \par Returning back to our application, we have to compute
    \begin{equation} \label{haar_average}
      \begin{split}
      {\mathlarger{\mathbb{E}}}_{U \in \mathscr{U}_N} P_E(U) &= 1 - \frac{1}{d^2} \Tr(T_{(12)(34)}^A \, {\mathlarger{\mathbb{E}}}_{U \in \mathscr{U}_N} {U^\dagger}^{\otimes 4} \, Q \, U^{\otimes 4}) = 1- \frac{1}{d^2} \sum_{\pi, \sigma \in \mathcal{S}_4} W_{\pi , \sigma} \Tr(T_{(12)(34)}^A \, T_\pi) \Tr(Q T_\sigma)\\
      &= 1 - \frac{1}{d^2} \sum_{\pi, \sigma \in \mathcal{S}_4} W_{\pi , \sigma} \Tr(T_{(12)(34)}^A \, T_{\pi}^A) \Tr(T_{\pi}^B) \Tr(Q T_\sigma) \\
      &= 1 - \frac{1}{d^2} \sum_{\pi, \sigma \in \mathcal{S}_4} W_{\pi , \sigma} \, \chi\left(T^A_{(12)(34) \cdot \pi} \right) \, \chi\left(T^B_\pi\right) \Tr(Q T_\sigma).
      \end{split}
    \end{equation}
    Notice that $\chi\left(T^A_{(12)(34) \cdot \pi}\right) = \sum_{\lambda \vdash q} d_\lambda^A \, \chi_\lambda\left((12)(34) \cdot \pi\right)$, where $d_\lambda^A$ is the dimension of the irreducible representation of $\mathscr{U}_{N_A}$ corresponding to $\lambda$ and similarly $\chi\left(T^B_\pi\right) = \sum_{\lambda \vdash q} d_\lambda^B \, \chi_\lambda(\pi)$, where $d_\lambda^B$ is the dimension of the irreducible representation of $\mathscr{U}_{N_B}$ corresponding to $\lambda$. All that is left is to compute $\Tr(Q T_\sigma)$. Since $Q= \frac{1}{d^2} \sum_{P \in \tilde{\mathscr{P}}_N} P^{\otimes 4}$ is permutationally invariant, $\Tr(Q T_\sigma)$ depends only on the conjugacy class of $\sigma$. Then, we have
    \begin{equation} \label{inner_prods}
    \begin{split}
    &\text{For } \sigma=(1)(2)(3)(4) \rightarrow \Tr(Q T_\sigma) = \Tr(Q)=d^2 \\
    &\text{For } \sigma=(1)(2)(34) \rightarrow \Tr(Q T_\sigma) = \Tr(Q T_{(34)}) = d \\
    &\text{For } \sigma=(12)(34) \rightarrow \Tr(Q T_\sigma) = \Tr(Q T_{(12)} T_{(34)}) = d^2 \\
    &\text{For } \sigma=(123)(4) \rightarrow \Tr(Q T_\sigma) = \Tr(Q T_{(123)}) = 1 \\
    &\text{For } \sigma=(1234) \rightarrow \Tr(Q T_\sigma)=\Tr(Q T_{(1234)}) = d.
    \end{split}
    \end{equation}
    Combining \cref{haar_average,weingarten,inner_prods} and substituting the irreducible characters of $\mathcal{S}_4$ and the dimensions of the irreducible representations of $\mathscr{U}_N$ in Mathematica, we obtain \cref{typical}.